\newif\ifextended
\extendedtrue

\newif\ifanonymous
\anonymousfalse

\documentclass[runningheads,envcountsame,envcountsect]{llncs}
\usepackage[T1]{fontenc}
\usepackage[british]{babel}
\usepackage[disable]{todonotes}
\usepackage{graphicx}

\usepackage[right]{lineno}

\usepackage[bookmarks,unicode,colorlinks,allcolors=blue,pdfusetitle]{hyperref}
\hypersetup{
  pdfauthor    = {Constantin Enea, Azadeh Farzan, and Dominik Klumpp},
  pdfsubject   = {\ifextended(Extended Version)\else CAV 2026\fi},
  colorlinks   = true, %
  urlcolor     = blue, %
  linkcolor    = blue, %
  citecolor    = blue  %
}

\usepackage{color}
\renewcommand\UrlFont{\color{blue}\rmfamily}
\usepackage[misc]{ifsym}
\usepackage{stmaryrd}
\usepackage{bbding}

\ifextended
  \usepackage[bibliography=common,forwardlinking=yes]{apxproof}
\else
  \usepackage[appendix=strip]{apxproof}
\fi

\usepackage{amsmath,amsthm,amsfonts,amssymb}
\usepackage{mathtools}
\usepackage{tikz}
\usetikzlibrary{shapes}
\usetikzlibrary{decorations.pathmorphing}

\usepackage[capitalise,nameinlink]{cleveref}

\usepackage{xspace}

\usepackage{multirow}
\usepackage{booktabs}

\usepackage{subcaption}
\usepackage{enumitem}

\usepackage{listings}
\numberwithin{equation}{section}

\newtheoremrep{definition}{Definition}[section]
\newtheoremrep{observation}[definition]{Observation}
\crefname{observation}{Observation}{Observations}
\newtheoremrep{lemma}[definition]{Lemma}
\newtheoremrep{proposition}[definition]{Proposition}
\newtheoremrep{theorem}[definition]{Theorem}
\newtheoremrep{corollary}[definition]{Corollary}

\newtheoremrep{thmalgorithm}[definition]{Algorithm}
\crefname{thmalgorithm}{Algorithm}{Algorithms}

\makeatletter
\def\orcidID#1{{\href{http://orcid.org/#1}{\protect\raisebox{-1.25pt}{\protect\includegraphics{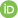}}}}}
\makeatother

\newcommand\N{\mathbb{N}}
\newcommand\powerset[1]{\wp\!\left(#1\right)}
\newcommand\compl[1]{{#1}^\mathsf{c}}

\newcommand\comp{\mathbin{\circ}}

\colorlet{th1}{yellow}
\colorlet{th2}{green!80!blue}
\tikzset{stmt/.style={font=\ttfamily\scriptsize,shape=rectangle,rounded corners=.2em,fill=gray!40,inner xsep=.3em,inner ysep=0em}}

\newcommand\ind[2]{\langle #1 : #2 \rangle}

\newcommand\Actions{\mathsf{Act}}

\newcommand\alphabet[1]{\Sigma(#1)}

\newcommand\indexed{\mathsf{idx}}

\newcommand\localTrace[2]{\lfloor #1 \rfloor_{#2}}

\newcommand\liftIdx[1]{\widehat{#1}}

\renewcommand\S{\mathbb{S}}
\newcommand\syncPred[1][]{\mathsf{sync}_{#1}}
\newcommand\sync[2][]{\syncPred[#1](#2)}

\newcommand\locks{\mathsf{lock}}
\newcommand\acq[1]{\texttt{acq}(#1)}
\newcommand\rel[1]{\texttt{rel}(#1)}

\newcommand\barrier{\bullet}

\newcommand\Loc{\mathbb{L}}
\newcommand\init{\mathsf{init}}
\newcommand\exit{\mathsf{exit}}

\newcommand\traces[1]{\mathcal{T}(#1)}

\newcommand\lang[1]{\mathcal{L}(#1)}

\newcommand\coveredby[1]{\sqsubseteq_{#1}}

\newcommand\F{\mathcal{F}}

\newcommand\fuseblocks[2]{#1 \lhd_\mathsf{at} #2}

\newcommand\progOrderRel{\mathsf{po}}
\newcommand\progOrder[2]{(#1,#2)\in\progOrderRel}

\newcommand\addbarriers[2]{#1 \lhd_\barrier #2}

\newcommand\phaseOrder{\mathsf{pho}}

\newcommand{\hurl}[1]{{\scriptsize\UrlFont\href{https://#1}{\path{#1}}}}

\newcommand\reminder[1]{}

\begin{document}
\title{On the Complexity of Checking Soundness\\ of Natural Reductions}
\ifextended\subtitle{(Extended Version)}\fi
\ifanonymous%
  \author{}
  \institute{}
\else
  \author{%
    Constantin Enea\inst{1}~\orcidID{0000-0003-2727-8865}\and%
    Azadeh Farzan\inst{2}~\orcidID{0000-0001-9005-2653}\and%
    Dominik Klumpp\inst{1}~\orcidID{0000-0003-4885-0728}%
  }%
  \authorrunning{C.~Enea, A.~Farzan, D.~Klumpp}
  \institute{%
    LIX -- CNRS -- \'Ecole Polytechnique, Palaiseau, France\\
    \email{\{cenea, klumpp\}@lix.polytechnique.fr}
    \and
    University of Toronto, Toronto, Canada\\
    \email{azadeh@cs.toronto.edu}
  }%
\fi

\maketitle              %
\begin{abstract}
The verification of \emph{reductions}, representative subsets of interleavings, simplifies correctness proofs of parameterized concurrent programs.
We introduce an expressive class of syntactic reductions, which we call \emph{natural reductions}.
Natural reductions are specified by introducing atomic blocks and global rendezvous points in the parameterized program's thread template.
We study the problem of deciding whether a given natural reduction is sound wrt.\ a given (semi-)com\-mutativity relation.
In the case that there is no synchronization between threads, we present a sound and complete polynomial-time algorithm.
In the case where synchronization is considered,
we provide a general lower bound for the problem (parametric in the synchronization mechanism),
and show that the problem is {\sc coNP}-hard already for a simple mechanism like locking.
\end{abstract}

\section{Introduction}

A {\em reduction} of a concurrent program is a program with a subset of the behaviors of the original that can be {\em soundly} verified in its place. A big body of research in program verification, in both algorithmic and deductive traditions, has argued that the use of reductions can be critical to the success of the verification task.  On the algorithmic verification side, the reductions either come from an a priori fixed set~\cite{pldi22:sound-seq,popl24:parameterized} or the verification algorithm has to pay the steep cost of searching for them~\cite{farzan20:red-safety-proofs}. 
On the deductive verification side, proof systems like CSPEC~\cite{DBLP:conf/osdi/ChajedKLZ18}, Armada~\cite{DBLP:conf/pldi/LorchCKPQSWZ20}, Anchor~\cite{flanagan:anchor} and Civl~\cite{DBLP:conf/cav/HawblitzelPQT15}
include proof tactics based on Lipton's reduction~\cite{lipton75:movers} for showing that a user-provided reduction is sound,
which can be interleaved with other proof tactics that relate to inductive invariant reasoning, for instance.
These proof systems have been used to verify real-world examples such as a concurrent garbage collector~\cite{DBLP:conf/cav/HawblitzelPQT15}, an implementation of the Paxos protocol~\cite{DBLP:conf/pldi/KraglEHMQ20}, or concurrent data structures~\cite{DBLP:journals/computing/MutluergilT19}.

In this paper, we investigate the complexity of checking whether some given reduction of a concurrent program is sound. We focus on {\em parameterized concurrent programs}, that is, concurrent programs in which an arbitrary number of threads can be instantiated from a single thread template. These programs appear everywhere in distributed and parallel computing. 
We consider a class of reductions that we call \emph{natural} that can be specified syntactically with ease and make it straightforward for programmers to envision the reduced program and reason about it. Hence, {\em natural reductions} enable programmers to interact with verification tools by proposing reductions that may simplify verification tasks.

Natural reductions are defined through two syntactic program transformations, which both have deep connections to traditions of reasoning about concurrent programs: 
\begin{itemize}
\item {\em Atomic Blocks}: The programmer can specify any number of syntactic atomic blocks in the program, which accordingly prunes the space of possible concurrent program behaviours,
  excluding all behaviours in which one thread's execution of the atomic block is interrupted by another thread.
\item {\em Synchronous Reductions:} The programmer uses a unique symbol as a {\em global rendezvous point} among all existing threads, and as such restricts the space of program behaviours to those in which the threads meet at the rendezvous point before continuing.
\end{itemize}
The former has been broadly used as the foundation of {\em local reasoning} in concurrency, and the latter has been used as the means of simplifying reasoning about distributed protocols \cite{DBLP:conf/podc/ChouG88,DBLP:journals/scp/ElradF82,DBLP:conf/cav/DamianDMW19}. In many distributed protocols, the behavior of each process is
structured as a sequence of rounds and it can be shown that processes executing rounds synchronously and in lock-step is a sound reduction of the original set of asynchronous behaviors (where processes may be executing different
rounds at a point of time). 
The user can use a combination of both to restrict the set of concurrent program behaviours and hence {\em specify} a reduction.  

We say a (specified) reduction is {\em sound} if every behaviour of the original program is equivalent to a behaviour in the reduction up to an {\em equivalence relation} induced by a {\em sound commutativity relation} on their common set of commands. A sound commutativity relation is a (not necessarily symmetric%
\footnote{In the non-symmetric case, the equivalence relation degenerates to a preorder.}%
) relation on the set of atomic actions that includes a pair $(a,b)$ if and only if the semantics of~$b\;a$ includes all behaviours of~$a\;b$.
A {\em sound reduction} can be soundly verified in place of the original program (i.e., the reduction satisfies the same postconditions as the original program).
Hence, given a {\em natural reduction} specified by the user, we are interested in the problem of determining whether it is sound.

In this paper, we investigate this problem for two different abstract models of program behaviours. First, we consider a model in which the program data is entirely abstracted away, and hence every valid syntactic behaviour of the program is valid behaviour in the abstract model. In a sense, the only information remaining about program semantics is given through the {\em sound commutativity} relation on the set of actions. We show that under this model, any natural reduction can be checked in {\em polynomial time} for soundness (\cref{thm:nat-red-soundness-polynomial}). Our argument is constructive. We present an algorithm that can check the soundness of a natural reduction proposed by atomic blocks in polynomial time. We also present an algorithm that can check the soundness of a synchronous reduction in polynomial time. We then argue how the algorithms compose, and therefore, a natural reduction with several atomic blocks and a rendezvous point can be checked for soundness in polynomial time.

Second, we consider a model in which the original program uses locks for synchronization, and the abstraction level at which the semantics of locking is visible and respected. This adds a layer of complication to the problem because one has to argue that every behaviour of the original program, {\em that respects the locking semantics}, has an equivalent up-to-commutativity behaviour in the specified natural reduction. We show that checking the soundness of a natural reduction in this setup is {\sc coNP}-hard, no matter how much we limit the use of constructs. In other words, it is  {\sc coNP}-hard even if a natural reduction is specified using a single atomic block, and it is {\sc coNP}-hard even if it is specified only using a rendezvous point.
Consequently, while fully respecting the locking semantics is theoretically more precise,
it is likely too costly for practical application in deductive proof systems.

\ifextended
Proofs of our key results can be found in the appendix.
\else
An extended version of our paper, including proofs, is available online~\cite{cav26:arxiv}.%
\fi

\section{Parameterized Programs and Reductions}

We use finite sequences to represent behaviours (\emph{interleavings}) of concurrent programs.
$X^*$ denotes the set of all finite sequences~$\sigma$ over elements in a given set~$X$ (the free monoid generated by~$X$).
$\powerset{X}$ is the powerset of~$X$.
We call a function $\mu : X \to \powerset{Y^*}$ a \emph{morphism}
and extend it to a function~$\mu : X^*\to\powerset{Y^*}$ (with $\mu(\varepsilon)=\{\varepsilon\}$, $\mu(x\sigma)=\mu(x)\mu(\sigma)$ for $x\in X,\sigma\in X^*$)
and further to a function~$\mu: \powerset{X^*}\to\powerset{Y^*}$ (with $\mu(L)=\bigcup_{\sigma\in L}\mu(\sigma)$).
By $\tau|_{Z}$, we denote the projection of~$\tau$ to~$Z$
(the morphism with $z|_{Z}=\{z\}$ for $z\in X\cap Z$ and $x|_{Z}=\{\varepsilon\}$ for $x\in X\setminus Z$).
The expression~$|\sigma|$ denotes the length of~$\sigma$, and $|\sigma|_x$~is the number of occurrences of~$x$ in~$\sigma$.
Square brackets $[x,y,z]$ denote multisets.

\subsection{Parameterized Programs}
A \emph{parameterized program} is a concurrent program,
where, at runtime, $n$~threads all execute the same thread template
(the number~$n$ is the \emph{parameter}).
For the purpose of this paper, we model atomic actions in the thread template abstractly as elements~$a,b,c,\ldots$ of an infinite set~$\Actions$.

Given a set of actions~$A$, we write $A_\indexed$ for $A \times \N$, the set of indexed actions,
where an \emph{indexed action} $\langle a : i \rangle\in A_\indexed$ represents action~$a$ being executed by the $i$-th thread
(the thread executing an action may matter, e.g.\ if the action accesses thread-local variables).
We represent the interleavings of actions from different threads by \emph{indexed traces}~$\tau\in A_\indexed^*$,
and write $\langle \tau :i \rangle$ for $\tau\in A^*$ to denote the corresponding indexed trace where all actions are executed by thread~$i$.
By $\localTrace{\tau}{i}$, we denote the trace of all actions~$a$ where $\langle a :i\rangle$ appears in~$\tau$
(the morphism with $\localTrace{\langle a:i\rangle}{i}=\{a\}$ and $\localTrace{\langle a:j\rangle}{i} = \{\varepsilon\}$ for $j\neq i$).
We lift a morphism~$\mu:X\to\powerset{Y^*}$ to indexed actions by defining the morphism~$\liftIdx{\mu}$ with $\liftIdx{\mu}(\langle a:i\rangle)=\{\,\langle \tau:i\rangle \mid \tau\in\mu(a)\,\}$.

Though we generally abstract away from specific actions and their semantics, %
we need to model the semantics of certain actions for the purpose of \emph{synchronization} between the threads,
through locks, semaphores, rendez-vous etc.
\begin{definition}
  A \emph{synchronization alphabet} $(\S, \syncPred)$ consists of a set~$\S$ of synchronization actions (disjoint from~$\Actions$)
  and a predicate~$\syncPred$ over traces in~$(\Actions \cup \S)_\indexed^*$,
  where $\syncPred$ is preserved by all permutations of thread indices.%
\end{definition}
In particular, the \emph{lock synchronization alphabet} over the infinite set~$M$ of lock variables is given by $\S_\locks = \{\, \acq{m}, \rel{m} \mid m \in M \,\}$
and the predicate $\syncPred[\locks]$ such that for all $\tau\in(\Actions \cup \S_\locks)_\indexed^*$, we have $\sync[\locks]{\tau}$ if and only if
\begin{multline}
   \forall m\in M\,.\, \tau|_{\{\acq{m},\rel{m}\}_\indexed} \in \{\,\langle \acq{m}\rel{m} :i \rangle \mid i \in \N \,\}^* \\\null\cdot \big (\{\varepsilon\} \cup \{\,\langle \acq{m}:i\rangle \mid i \in \N \,\} \big)
\end{multline}
The \emph{trivial synchronization alphabet} is $(\emptyset,\syncPred[\top])$ where $\sync[\top]{\tau}$ holds for all~$\tau$.

The thread template of a parameterized program is a control flow graph.
\begin{definition}
Given a synchronization alphabet $(\S,\syncPred)$,
a \emph{control flow graph} $G=(\Loc, \Delta, \ell_\init, \ell_\exit)$  consists of
a finite set of locations~$\Loc$,
a finite transition relation $\Delta \subseteq \Loc \times (\Actions\cup\S) \times \Loc$,
an initial location~$\ell_\init \in \Loc$ and an exit location~$\ell_\exit\in\Loc$, with $\ell_\init \neq \ell_\exit$.
\end{definition}
We assume that every location is reachable from~$\ell_\init$,
and $\ell_\exit$~is reachable from every location.
Furthermore, for each~$a\in\Actions$, there exists at most one edge in~$\Delta$ labeled by~$a$
(this assumption is only for notational convenience).
We denote the set of actions occurring in~$G$ by~$\alphabet{G}$ (i.e., $\alphabet{G} \subseteq \Actions\cup\S$).
We write $\ell \xrightarrow{\tau}_\Delta \ell'$ for $\ell,\ell'\in\Loc$ and $\tau\in \alphabet{G}^*$ if there is a path from~$\ell$ to~$\ell'$ labeled by~$\tau$.
$\traces{G}$ is the language of all traces~$\tau\in \alphabet{G}^*$ such that $\ell_\init \xrightarrow{\tau}_\Delta \ell_\exit$ holds.
\begin{definition}
  A \emph{parameterized program} $P = ((\S,\syncPred), G)$ consists of
  a synchronization alphabet~$(\S,\syncPred)$
  and a control flow graph~$G$ (the \emph{thread template}).
\end{definition}
The interleavings of a program~$P = ((\S,\syncPred), G)$ are given by
\begin{multline}
  \lang{P} = \{\, \tau \in \Actions_\indexed^* \mid \exists \hat{\tau}\in(\Actions\cup\S)_\indexed^* \,.\, \tau = \hat{\tau}|_{\Actions_\indexed}
    \text{ and } \sync{\hat{\tau}}\\
    \text{ and } \forall i\in\N\,.\,
      \localTrace{\hat{\tau}}{i} \in \traces{G} \cup \{\varepsilon\}
     \,\}
\end{multline}
Every interleaving~$\tau$ corresponds to a trace~$\hat{\tau}\in(\Actions\cup\S)_\indexed^*$
which follows the synchronization discipline (e.g., no two threads hold the same lock at the same time),
and in which every thread follows a path from the initial to the exit location (or does not run at all).
In the interleaving~$\tau$ itself, the synchronization actions are dropped; their only purpose is to enforce the synchronization discipline.
As every thread must reach the exit location,
this language of interleavings is suitable for the verification of \emph{postconditions}.

\subsection{Reductions}
The central idea behind \emph{commutativity} and \emph{Mazurkiewicz equivalence}~\cite{mazurkiewicz:trace-theory} is that certain actions (of different threads) in an interleaving can be swapped without changing its behaviour.
This idea generalizes to \emph{semi-commutativity}~\cite{clerbout:semi-commutativity}, in which some pairs of actions can be swapped in one direction (yielding a superset of behaviours) but not vice versa.
As a result, Mazurkiewicz equivalence degenerates to a preorder (the relation is not necessarily symmetric).

For a program~$P=((\S,\syncPred),G)$,
we presume a given (semi-)commutativity relation $I\subseteq (\alphabet{G}\cap\Actions)^2$, not necessarily symmetric,
which defines the actions that may be swapped if performed by different threads
(actions from the same thread can never be reordered).
Formally, the \emph{covering preorder}~${\coveredby{I}}$ is the smallest reflexive-transitive relation over $\Actions_\indexed^*$
with $\rho \,\langle a:i\rangle\langle b:j\rangle \,\sigma \coveredby{I} \rho\, \langle b:j\rangle\langle a:i\rangle\, \sigma$ for all $\rho,\sigma\in\Actions_\indexed^*$, $(a,b)\in I$ and $i\neq j\in\N$.
\begin{definition}
Given $L_1,L_2 \subseteq \Actions_\indexed^*$,
$L_1$~is a \emph{Mazurkiewicz reduction} of~$L_2$
if $L_1 \subseteq L_2$ holds,
and for all $\tau\in L_2$,
there exists some $\tau'\in L_1$ with $\tau \coveredby{I} \tau'$.
\end{definition}

Note that $I$~is a relation over actions in~$\Actions$ only; we do not consider commutativity of synchronization actions in~$\S$.
The reason is that $I$~is an abstraction of the concrete semantics of actions (which we do not model).
In particular, it can be thought of as the information which actions can be swapped without affecting the semantics.
By contrast, we explicitly model the semantics of synchronization actions with the predicate~$\syncPred$,
and only consider interleavings that obey the semantics.
Hence, there is no need to abstract this semantics with~$I$.

\section{Natural Reductions}
This section formally defines our notion of \emph{natural reductions}
based on the syntactic introduction of \emph{atomic blocks} and \emph{global rendez-vous points} in a thread template.
The definitions here describe the permissible syntactic changes;
they do not guarantee by construction that the reduced program soundly represents the original program.
Rather, our definitions specify permissible inputs for the problem of \emph{deciding} the soundness.
We discuss the corresponding decision problems (as well as concrete algorithms) in the subsequent sections.

In the following, and for the remainder of the paper, we fix a parameterized program~$P=((\S,\syncPred), G)$
and a commutativity relation~$I \subseteq (\alphabet{G}\cap \Actions)^2$.

\subsection{Atomic Blocks}
The following definition captures the syntactic introduction of an atomic block in our control flow graph-based formalism:
\begin{definition}
  An \emph{atomic fusion} $\F = (G', (\beta_1, G_{\beta_1}), \ldots, (\beta_n, G_{\beta_n}))$ for~$P$ consists of a control flow graph~$G'$,
  distinct actions~$\beta_1,\ldots,\beta_n\in\alphabet{G'}$,
  and control flow graphs $G_{\beta_1},\ldots,G_{\beta_n}$,
  where for all $k=1,\ldots,n$, we have
  \begin{itemize}
    \item $\alphabet{G_{\beta_k}} \subseteq \Actions\setminus\{\beta_1,\ldots,\beta_n\}$, i.e., $G_{\beta_k}$ contains neither synchronization actions nor any~$\beta_{k'}$,
    \item $\traces{G_{\beta_k}} \neq \emptyset$,\,i.e.,\,$G_{\beta_k}$\,has at least one path from the initial to the exit location,
    \item $G=G'[\beta_1\mapsto G_{\beta_1},\ldots,\beta_n\mapsto G_{\beta_n}]$, i.e., $G$ can be derived from~$G'$ by inserting each~$G_{\beta_k}$ in place of the (unique) edge labeled by the respective~$\beta_k$.
  \end{itemize}
\end{definition}
In the definition above, $G'$ is a new thread template.
Each \emph{block symbol}~$\beta_k$ represents an (atomic) execution of an atomic block,
whereas $G_{\beta_k}$ describes the possible paths through the body of the atomic block.
There may be multiple paths through an atomic block, due to branching and loops.

\begin{corollary}
  Let $\F = (G',(\beta_1, G_{\beta_1}), \ldots, (\beta_n, G_{\beta_n}))$ be an atomic fusion for~$P$,
  and $\mu_\F: \Actions\cup\S \to \powerset{(\Actions \cup \S)^*}$ the morphism with $\mu_\F(\beta_k) = \traces{G_{\beta_k}}$ for all~$k$,
  and $\mu_\F(a) =\{a\}$ for all other $a\in\Actions\cup\S$.
  It holds that $\traces{G} = \mu_\F(\traces{G'})$.
\end{corollary}
We call~$\mu_\F$ the \emph{fusion morphism} and make use of it throughout the paper.
\smallskip

An atomic fusion induces a new program~$\fuseblocks{P}{\F} := ((\S,\syncPred), G')$.
The interleavings of~$\fuseblocks{P}{\F}$ correspond to interleavings of~$P$ where the atomic block is executed without interruption by other threads
(the precise path through the atomic block is abstracted by a block symbol~$\beta_k$).
Hence, $\liftIdx{\mu_\F}(\lang{\fuseblocks{P}{\F}})$ always contains a subset of the interleavings in~$\lang{P}$.
However, it is not a~priori clear that this subset of interleavings is representative, i.e., forms a sound reduction.
\begin{definition}
  An atomic fusion~$\F$ of~$P$ is \emph{sound} (wrt.~$I$)
  if $\liftIdx{\mu_\F}(\lang{\fuseblocks{P}{\F}})$ is a Mazurkiewicz reduction (up to~$I$) of~$\lang{P}$.
\end{definition}
\Cref{sec:check-atomic-blocks} discusses the problem of \emph{deciding} if an atomic fusion is sound.

\subsection{Synchronous Reductions}
To support synchronous reductions, that for instance align multiple \emph{rounds} performed by threads,
we introduce a notion of \emph{global rendez-vous points}, or \emph{sync-points} for short.
When a thread encounters a sync-point, it waits until all other threads also reached a sync-point.
Only then, all threads can pass the sync-point together;
afterwards, each thread continues with its respective computation.
The only exception is that a thread that has already reached its exit location need not participate in sync-points (the still-running threads can pass their sync-points without it).

We formalize sync-points as a synchronization alphabet $(\S_\barrier, \syncPred[\barrier])$,
with a unique \emph{sync-point symbol}~$\barrier$ (i.e., $\S_\barrier = \{\barrier\}$).
To define the synchronization predicate, given a finite set $T=\{t_1,\ldots,t_n\}\subseteq \N$,
we let $\ind{\barrier}{T}$ denote the sequence $\ind{\barrier}{t_1}\ldots\ind{\barrier}{t_n}$.
We then define inductively, for all finite sets~$T\subseteq\N$:
\begin{align}
  \sync[\barrier,T]{\tau} &\null \iff \begin{aligned}[t]
    \tau=\tau_1\tau_2 \land\null &\tau_1\in\big((\Actions\times T)^*+\ind{\barrier}{T}\big)^* \\
                 \null\land\null &\exists T'\subsetneq T\,.\, \sync[\barrier,T']{\tau_2}
  \end{aligned}\\
  \sync[\barrier]{\tau}   &\null \iff \exists T\subseteq_{\mathsf{fin}}\N\,.\, \sync[\barrier,T]{\tau}
  \label{eq:sync-barrier}
\end{align}
The definition of $\syncPred[\barrier]$ states that whenever some set of threads~$T$ are running,
they can only pass a sync-point~$\barrier$ together (as in the sequence $\ind{\barrier}{T}$).
At any point, some threads may terminate (i.e., not perform any further actions);
from thereon only the remaining threads~$T'$ need to synchronize on sync-points.

As sync-points may be added to programs that already use other synchronization actions,
we note that synchronization alphabets $(\S_1, \syncPred[1]), (\S_2,\syncPred[2])$,
with $\S_1 \cap \S_2 =\emptyset$,
can be combined into a joint synchronization alphabet $(\S_1,\syncPred[1]) \oplus (\S_2,\syncPred[2]) := (\S_1\cup\S_2, \syncPred[12])$,
where $\sync[12]{\tau}$ holds if and only if $\sync[1]{\tau|_{(\Sigma \cup \S_1)_\indexed}}$ and $\sync[2]{\tau|_{(\Sigma\cup\S_2)_\indexed}}$ hold.
We define sync-point instrumentation as follows:
\begin{definition}
  \label{def:barrier-instr}
  A \emph{sync-point instrumentation} of~$P$ is a control flow graph~$G'$ over $\Actions \cup \S \cup \S_\barrier$
  with $\traces{G} = \traces{G'}|_{\Actions \cup \S}$,
  and for all $\tau_1,\tau_2\in \traces{G'}$, we have that $\tau_1|_\Actions = \tau_2|_\Actions$ implies $\tau_1=\tau_2$.
\end{definition}
A sync-point instrumentation is thus a modified thread template that inserts the sync-point symbol~$\barrier$ at various locations.
The injectivity condition for projection means that for every~$\tau\in \traces{G}$,
there is a \emph{unique} way to insert sync-points in it and derive a trace in~$\traces{G'}$.
One way to ensure this injectivity syntactically
is to construct~$G'$ from~$G=(\Loc,\Delta,\ell_\init,\ell_\exit)$ by inserting sync-points at a set of locations~$M\subseteq \Loc$.
I.e., we define $G'=(\Loc\uplus\{\hat{\ell}\mid\ell\in M\}, \Delta',\ell_\init,\ell_\exit)$
with
\begin{multline*}
\Delta' = \{\,(\ell,a,\ell')\in\Delta \mid \ell\notin M\,\} \cup \{\,(\ell,\barrier,\hat{\ell}) \mid \ell\in M\,\} \\\null\cup \{\,(\hat{\ell},a,\ell') \mid \ell\in M \land (\ell,a,\ell')\in\Delta \,\}\ .
\end{multline*}
We add a copy~$\hat{\ell}$ of every location~$\ell\in M$,
such that $\ell$~retains its incoming edges,
$\ell$ and $\hat{\ell}$ are connected by a sync-point,
and all outgoing edges of~$\ell$ are moved to~$\hat{\ell}$ (i.e., after the sync-point).
This syntactic sync-point insertion at locations~$M$ is sufficient to specify reductions e.g.\ for distributed protocols based on \emph{rounds},
as the point where a thread advances from one round to another can be identified syntactically in the code.
\begin{observation}
  Any~$G'$ constructed as described above is a sync-point instrumentation of~$P$.
\end{observation}
A given sync-point instrumentation~$G'$ induces a new program~$\addbarriers{P}{G'}:=((\S_{+\barrier}, \syncPred[+\barrier]), G')$ with the combined synchronization alphabet $(\S_{+\barrier},\syncPred[+\barrier]) := (\S, \syncPred) \oplus (\S_\barrier, \syncPred[\barrier])$.
The additional synchronization enforced by the sync-points yields a subset of interleavings: $\lang{\addbarriers{P}{G'}}\subseteq \lang{P}$.
\begin{definition}
  A sync-point instrumentation~$G'$ of~$P$ is \emph{sound} (with respect to~$I$)
  if $\lang{\addbarriers{P}{G'}}$ is a Mazurkiewicz reduction (up to~$I$) of $\lang{P}$.
\end{definition}
\Cref{sec:check-barriers} discusses the problem of \emph{deciding} whether a given sync-point instrumentation is sound.

\subsection{Combining Atomic Blocks and Synchronous Reductions}
We combine atomic blocks and synchronous reductions in a straightforward manner
in order to define the class of~\emph{natural reductions}.
\begin{definition}
  A \emph{natural reduction specification}~$(\F, G')$ for~$P$
  consists of an atomic fusion~$\F$ of~$P$
  and a sync-point instrumentation~$G'$ of~$\fuseblocks{P}{\F}$.
\end{definition}
Note that by this definition,
atomic blocks may not contain sync-points.
There would be no point for a thread to wait on all other threads at a sync-point,
while simultaneously preventing other threads from making progress and reaching sync-points (as that would interrupt the atomic block).

The notion of soundness follows directly from soundness of atomic fusions and sync-point instrumentations:
\begin{definition}
  A natural reduction specification~$(\F,G')$ for~$P$ is \emph{sound} (wrt.~$I$)
  if $\liftIdx{\mu_\F}(\lang{\addbarriers{\fuseblocks{P}{\F}}{G'}})$ is a Mazurkiewicz reduction (up to~$I$) of~$\lang{P}$.
\end{definition}

\section{Deciding the Soundness Problem}
\label{sec:deciding-soundness}

We first consider the problems of deciding soundness for atomic fusions and sync-point instrumentations separately,
in~\cref{sec:check-atomic-blocks} and \cref{sec:check-barriers}, respectively.
\Cref{sec:check-natural-reductions} then proposes a combined algorithm for deciding soundness of natural reduction specifications.

\subsection{Deciding Soundness of Atomic Blocks}
\label{sec:check-atomic-blocks}
We turn to the problem of deciding whether a given atomic fusion is sound.
The complexity of this problem, like many problems over parameterized programs, depends crucially on the mechanisms allowed for synchronization between threads~\cite{esparza:keeping-crowd-safe}.
In fact, we link its complexity to the \emph{coverability problem} for configurations of arbitrary width.

We begin by examining the notion of a \emph{sound} atomic fusion~$\F$ of~$P$ more closely.
In essence, it states that every interleaving of the original program~$P$
is covered by some interleaving in which the actions inside an atomic block execute without interruption by other threads.
To ensure this, we must check that each action that may be interleaved inside the atomic block
can be swapped either before all actions of the atomic block (it \emph{escapes to the left}) or after all actions of the atomic block (it \emph{escapes to the right}).%
\footnote{%
  This is a shift from the classical perspective~\cite{lipton75:movers}, which reasons about \emph{moving} the actions \emph{of the atomic block}.
  We discuss limitations of that perspective in~\cref{sec:mover-incomplete}.
}
\begin{figure}[t]
\begin{subfigure}{0.38\textwidth}
\centering
\begin{tikzpicture}[baseline=(b),scale=0.85]
  \node[] (bl) {$\ind{... z_i ...}{t}$};
  \node[right=1mm of bl,inner sep=0] (b) {$\ind{b}{t'}$};
  \node[right=1mm of b] (br) {$\ind{... z_{j} ...}{t}$};
  \path[draw] (bl.north east) -- (bl.north west) -- (bl.south west) -- (bl.south east) decorate[decoration={zigzag,segment length=0.7mm,amplitude=0.3mm}]{ -- (bl.north east)};
  \path[draw] (br.north west) -- (br.north east) -- (br.south east) -- (br.south west) decorate[decoration={zigzag,segment length=0.7mm,amplitude=0.3mm}]{ -- (br.north west)};

  \draw[thick,densely dashed,blue]
    (bl.north west) ++(left:2mm) edge[<-] ++(up:5mm)
    ++(up:5mm) -| node[pos=0.175,red](bl-b){\scriptsize\XSolidBold} (b.north);
  \draw[densely dotted,red,thick] (bl.north) ++(left:2mm) |- (bl-b.center);

  \draw[thick,densely dashed,blue]
    (br.south east) ++(right:2mm) edge[<-] ++(down:5mm)
    ++(down:5mm) -| node[pos=0.245,red](br-b){\scriptsize\XSolidBold} (b.south);
  \draw[densely dotted,red,thick] (br.south) ++(left:2mm) |- (br-b.center);
\end{tikzpicture}
\caption{$(z_i,b)\notin I,(b,z_j)\notin I$.}
\label{fig:escape-simple}
\end{subfigure}
\hfill
\begin{subfigure}{0.65\textwidth}
\centering
\begin{tikzpicture}[baseline=(a),scale=0.85]
  \node[] (bl) {$\ind{... z_i ...}{t}$};
  \node[right=1mm of bl,inner sep=0] (a) {$\ind{a}{t_1}$};
  \node[right=0mm of a,inner sep=0]  (b) {$\ind{b}{t_1}$};
  \node[right=0mm of b,inner sep=0]  (c) {$\ind{c}{t_2}$};
  \node[right=1mm of c] (br) {$\ind{... z_{j} ...}{t}$};
  \path[draw] (bl.north east) -- (bl.north west) -- (bl.south west) -- (bl.south east) decorate[decoration={zigzag,segment length=0.7mm,amplitude=0.3mm}]{ -- (bl.north east)};
  \path[draw] (br.north west) -- (br.north east) -- (br.south east) -- (br.south west) decorate[decoration={zigzag,segment length=0.7mm,amplitude=0.3mm}]{ -- (br.north west)};

  \draw[thick,densely dashed,blue]
    (bl.north west) ++(left:2mm) edge[<-] ++(up:3mm)
    ++(up:3mm) -| node[pos=0.165,red](bl-a){\scriptsize\XSolidBold} (a.north);
  \draw[densely dotted,red,thick] (bl.north) ++(left:2.5mm) |- (bl-a.center);

  \draw[thick,densely dashed,blue]
    (br.south east) ++(right:2mm) edge[<-] ++(down:3mm)
    ++(down:3mm) -| node[pos=0.25,red](br-c){\scriptsize\XSolidBold} (c.south);
  \draw[densely dotted,red,thick] (br.south) ++(left:2.5mm) |- (br-c.center);

  \draw[thick,densely dashed,blue]
    (bl.north west) ++(left:4mm) edge[<-] ++(up:5mm)
    ++(up:5mm) -| node[pos=0.39,red](a-b){\scriptsize\XSolidBold} (b.north);
  \draw[densely dotted,red,thick] (a.north) ++(right:2mm) |- (a-b.center);

  \draw[thick,densely dashed,blue]
    (br.south east) ++(right:4mm) edge[<-] ++(down:5mm)
    ++(down:5mm) -| node[pos=0.4025,red](b-c){\scriptsize\XSolidBold} (b.south);
  \draw[densely dotted,red,thick] (c.south) ++(left:3mm) |- (b-c.center);
\end{tikzpicture}
\caption{$(z_i,a)\notin I, \progOrder{a}{b}, (b,c)\notin I, (c,z_j)\notin I$}
\label{fig:escape-complex}
\end{subfigure}
\caption{Action~$b$ cannot escape to the left/right of the atomic block $z_1\ldots z_m$.}
\label{fig:escape}
\end{figure}
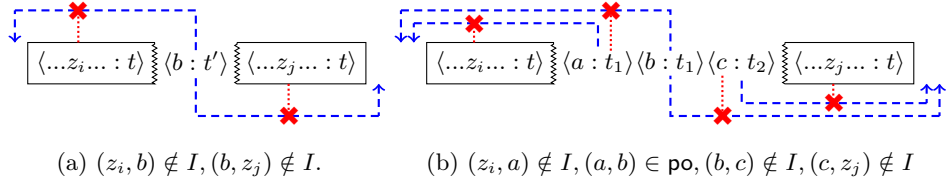
There can be different obstacles that prevent an action from escaping an atomic block~$z_1\ldots z_m$,
thus causing the atomic fusion to be unsound.
In the simplest case, illustrated in~\cref{fig:escape-simple},
an action~$b$ does not commute to the left of some~$z_i$ nor to the right of some~$z_j$, with $i<j$.
Then, the kind of interleaving shown in~\cref{fig:escape-simple} has no representative where $z_1\ldots z_m$ executes atomically,
and the given atomic fusion is unsound.

However, the problem can be more complex.
For instance, in~\cref{fig:escape-complex},
$b$~commutes to the left of all $z_1,\ldots,z_m$,
but a previous action~$a$ of the same thread does not commute with~$z_i$.
As we cannot swap $\ind{b}{t_1}$ with $\ind{a}{t_1}$, in this scenario, $\ind{b}{t_1}$ cannot escape to the left.
On the other hand, $\ind{b}{t_1}$ cannot escape to the right either,
due to the presence of an action~$\ind{c}{t_2}$ such that $(b,c)\notin I$.
In order for $\ind{b}{t_1}$ to escape to the right, $\ind{c}{t_2}$ would also have to escape \emph{to the right}, but we have~$(c,z_j)\notin I$.

As shown here, even if individual actions interleaved in an atomic block can escape,
in general we have to reason about complex chains of dependencies (e.g., if $\ind{a}{t_1}$ can only escape to the right, then $\ind{b}{t_1}$ must also escape to the right, and hence the same applies to~$\ind{c}{t_2}$).
Whether such a chain of dependencies causes unsoundness however crucially depends on
whether the actions~$\ind{a}{t_1},\ind{b}{t_1},\ind{c}{t_2}$ can actually be interleaved inside the block in this manner,
or whether synchronization prevents this.
For instance, if actions~$b$ and~$c$ are protected by the same lock,
the issue above could not occur
($t_1$~and~$t_2$ cannot both hold the lock at the same time).
The soundness of atomic fusions is thus tightly linked to the question of \emph{coverability} under synchronization:
\begin{definition}[Coverability]
  Let $C=[\ell_1,\ldots,\ell_n]$ be a finite multiset of locations of~$P$.
  The configuration~$C$ is \emph{coverable} in~$P$
  if there exist a trace $\hat{\tau}\in(\alphabet{G}\cup\S)_\indexed^*$
  with $\sync{\hat{\tau}}$,
  and (distinct) $t_1,\ldots,t_n\in\N$ with $\ell_\init \xrightarrow{\localTrace{\hat{\tau}}{i}}_\Delta \ell_i$.
\end{definition}

\begin{theoremrep}
  \label{thm:atomic-blocks-coverability}
  For every synchronization alphabet,
  the coverability problem over programs with this synchronization alphabet
  is polynomial-time reducible to the problem of deciding whether a given atomic fusion of a given program  is unsound wrt.\ a given commutativity relation.
\end{theoremrep}
\begin{proofsketch}
  Given configuration $C=[\ell_1,\ldots,\ell_n]$,
  we let $n$~fresh actions $a_1,\ldots,a_n$ form a dependency chain with actions $b_1,b_2$ belonging to an atomic block
  (i.e., $(b_1,a_1)\notin I$, $(a_i,a_{i+1})\notin I$, $(a_n,b_2)\notin I$);
  all other actions commute.
  We insert the actions~$a_i$ at the corresponding locations~$\ell_i$, and add the atomic block~$b_1 b_2$ as a branch from the initial to the exit location.
  The atomic fusion is unsound if and only if a trace containing $\langle b_1:i_0 \rangle \langle a_1:i_1 \rangle \ldots \langle a_n:i_n\rangle \langle b_2:i_0\rangle$ is allowed by $\syncPred$,
  which holds if and only if $C$~is coverable.
\end{proofsketch}
\begin{proof}
  
Given a program~$P=((\S,\syncPred), G)$ and a configuration~$C=[\ell_1,\ldots,\ell_r]$,
we construct a program $P'=((\S,\syncPred), G')$, an atomic fusion $\F=(G'',(\beta,G_\beta))$ of~$P'$, and a commutativity relation~$I\subseteq \alphabet{G'}^2$
such that $\F$~is \emph{un}sound if and only if~$C$ is coverable in~$P$.

\paragraph{Construction.}
Let $a,b,c_1,\ldots,c_r$ be $r+2$ fresh statements (i.e., $a,b,c_1,\ldots,c_r\notin \alphabet{G}$),
and let $\Sigma' := \alphabet{G} \cup \{a,b,c_1,\ldots,c_r\}$.
We define a commutativity relation~$I \subseteq \Sigma' \times \Sigma'$,
by setting
\[
  I := (\Sigma' \times \Sigma') \setminus \{(a,c_1),(c_1,c_2),\ldots,(c_{r-1},c_r),(c_r,b)\}\ .
\]
Further, let $G'$ be the thread template $G$,
where
\begin{itemize}
\item two fresh locations $\ell_a$ and $\ell_{ab}$ have been added, and $\ell_{ab}$ is the exit location of~$G'$,
\item as well as edges $(\ell_\init, a, \ell_a)$, $(\ell_a, b, \ell_{ab})$,
\item and edges $(\ell_i,c_i,\ell_{ab})$ for all $i=1,\ldots,r$.
\end{itemize}
Finally, let $G''$ be the thread template $G'$,
where edges $(\ell_\init, a, \ell_a)$ and $(\ell_a, b, \ell_{ab})$ have been replaced by a single edge from $\ell_\init$ to $\ell_{ab}$ labeled by the atomic block symbol~$\beta$,
and let $G_\beta = (\{\ell_\init,\ell_a,\ell_{ab}\}, \{(\ell_\init, a, \ell_a),(\ell_a, b, \ell_{ab})\},\ell_\init, \ell_{ab})$.
Then $\F$~is indeed an atomic fusion of~$P'$.

\paragraph{Reasoning.}
We claim that $C$~is coverable in~$P$ if and only if $\F$~is unsound wrt.~$I$.
To see this,
assume that $C$~is coverable, and let $\hat{\tau}\in(\alphabet{G}\cup\S)^*$~be the corresponding trace as in the definition of coverability.
Then $\tau|_{\Actions_\indexed}\,a\,c_1\ldots c_r\,b$ is an interleaving of~$P'$,
which has no representative in~$\fuseblocks{P'}{\F}$.

Conversely, assume that $\F$ is unsound,
i.e., there exists some trace of~$P'$ that has no representative in~$\liftIdx{\mu_\F}(\lang{\fuseblocks{P}{\F}})$.
As almost all statements commute, this can only be possible if there exists a trace of the form $\tau' = \tau a\,c_1\ldots c_r\,b\sigma$ in~$\lang{P'}$.
Let $\hat{\tau}'\in (\alphabet{G}\cup\S)^*$ be the corresponding trace with $\sync{\hat{\tau}'}$, as in the definition of~$\lang{P'}$.
The prefix $\hat{\tau}''$ of $\hat{\tau}'$ with $\hat{\tau}''|_{\Actions_\indexed} = \tau$ witnesses the coverability of~$C$.
\qed

\end{proof}
In particular, the width of the configurations in the coverability problem corresponds to the size of simple cycles in the complement of~$I$.

As a corollary of~\cref{thm:atomic-blocks-coverability},
deciding soundness of atomic fusions for programs with powerful synchronization mechanisms such as \emph{broadcast} has non-elementary complexity~\cite{esparza:keeping-crowd-safe}.
Even for programs with simpler synchronization mechanisms, such as boolean variables,
the problem is {\sc coNP}-hard.
This motivates us to consider an abstract view of the program,
which ignores the semantics of such synchronization actions, and treats them like any other ordinary actions.
This yields an over-approximation of the program,
which contains some interleavings that would be ruled out by the synchronization mechanism.
Reasoning on this abstract view of the program allows us to \emph{certify} the soundness of atomic fusions (and later, of natural reductions generally):
any atomic fusion that is sound wrt.\ the abstract view is also sound wrt.\ a more concrete view.
A determination that an atomic fusion is \emph{unsound} does not transfer to the concrete level:
the perceived unsoundness may be an artifact of the abstraction.
Yet, the abstract view provides for an efficient and predictable decision procedure.
\medskip

Consequently, we assume for the remainder of this section that our program~$P$ uses the trivial synchronization alphabet~$(\emptyset,\syncPred[\top])$.
Furthermore, let $\F = (G', (\beta_1, G_{\beta_1}), \ldots, (\beta_n, G_{\beta_n}))$ be an atomic fusion of~$P$.
We present an algorithm to decide whether the given atomic fusion~$\F$ is sound wrt.~$I$.

The algorithm is based on detecting the different kinds of ``obstacles'' that might prevent an action interleaved inside an atomic block from escaping,
as discussed above (and illustrated in~\cref{fig:escape}).
\begin{toappendix}
For the purpose of the proof of~\cref{conj:polynomial-sound-atomic} below, we introduce an extended (infinite) alphabet $\Sigma' := \alphabet{G} \cup \{\,\beta_k^\tau \mid k=1,\ldots,n \text{ and } \tau\in\traces{G_{\beta_k}} \,\}$.
We define the commutativity relation~$I'\subseteq \Sigma'\times\Sigma'$ below, where $a,b$ range over~$\alphabet{G}$:
\begin{align*}
  I' := I
  &\null \cup \{\, (a, \beta_k^\tau) \mid \tau=z_1\ldots z_m \land \forall i=1,\ldots,m \,.\, (a, z_i) \in I \,\}\\
  &\null \cup \{\, (\beta_k^\tau, b) \mid \tau=z_1\ldots z_m \land \forall i=1,\ldots,m \,.\, (z_i, b) \in I \,\}\\
  &\null \cup \begin{aligned}[t]
    \{\, (\beta_{k_1}^{\tau_1}, \beta_{k_2}^{\tau_2}) \mid\null &\tau_1=z_1^1\ldots z_{m_1}^1, \tau_2=z_1^2\ldots z_{m_2}^2 \\
      &\null \land \forall i=1,\ldots,m_1\,.\, \forall j=1,\ldots,m_2\,.\, (z_i^1, z_j^2) \in I \,\}
    \end{aligned}
\end{align*}
\end{toappendix}
For $a,b\in\alphabet{G}$,
we write $\progOrder{a}{b}$ (\emph{program order}) if there exists a trace of the form $\rho\iota\sigma\in \traces{G}$, where $\iota$ begins with~$a$ and ends with~$b$
(possibly~$a=\iota=b$).
$\compl{I} := \alphabet{G}^2 \setminus I$ denotes the complement of~$I$.
\newcommand\atRel{\mathsf{at}}%
\newcommand\at[2]{(#1,#2)\in \atRel}%
We write $\at{a}{b}$ if there exists $\rho b \iota a \sigma\in\traces{G_{\beta_k}}$, for some $\beta_k$ (the reverse program order within atomic sections).
We define~${\rightarrowtail}$
as the relation $\compl{I} \comp \big( (\progOrderRel \cup \atRel) \comp \compl{I}\big)^+$,
where ${\comp}$ denotes relational composition.
The following result formalizes the idea that soundness of~$\F$
corresponds to the non-existence of the kind of dependence chains $\ind{z_i}{t},\ind{a}{t_1},\ind{b}{t_1},\ind{c}{t_2},\ind{z_j}{t}$
as in the example of~\cref{fig:escape-complex}.
\begin{toappendix}
  We say that an atomic block~$\beta_k$ occurs non-atomically in an interleaving~$\tau\in\Actions_\indexed^*$
  if there exists $z_1\ldots z_m\in\traces{G_{\beta_k}}$ and $i\in\N$
  such that $\localTrace{\tau}{i}=z_1\ldots z_m$, but $\ind{z_1\ldots z_m}{i}$ is not an infix in~$\tau$.
\end{toappendix}
\begin{propositionrep}
  \label{conj:polynomial-sound-atomic}
  The atomic fusion~$\F$ is sound
  if and only if
  there do not exist $z_1\ldots z_m\in \traces{G_{\beta_k}}$ (for some~$k$) and $i<j\in\{1,\ldots,m\}$
  such that
  $z_i \rightarrowtail z_j$ holds.
\end{propositionrep}
\begin{proofsketch}
  If $z_i\rightarrowtail z_j$ holds, we construct an interleaving in which $z_1\ldots z_m$ cannot be made atomic,
  except possibly by breaking apart another atomic block.
  Conversely, under the assumption that no $z_1\ldots z_m$ with $z_i\rightarrowtail z_j$ exists,
  we show that every interleaving can be transformed into one where all atomic blocks are represented by their respective block symbol~$\beta_k$.
\end{proofsketch}
\begin{proof}
  We show the two directions separately.
  \medskip

  \textbf{Completeness.}
  If there exist a trace $z_1\ldots z_m\in \traces{G_{\beta_k}}$ (for some~$k$) and indices $i<j\in\{1,\ldots,m\}$
  such that $z_i\rightarrowtail z_j$ holds,
  we construct an interleaving of the form~$\tau = \rho\,\ind{z_1\ldots z_i}{1}\,\iota\,\ind{z_{i+1}\ldots z_m}{1}\,\sigma \in \lang{P}$,
  where $\iota$~does not contain any steps of thread~$1$,
  such that $\liftIdx{\mu_\F}(\lang{\fuseblocks{P}{\F}})$ cannot contain any interleaving covering~$\tau$.
  The prefix~$\rho \in \Actions_\indexed^*$ brings each involved thread to a suitable starting point,
  and the suffix~$\sigma$~completes the execution of each thread.

  As we have $z_i \rightarrowtail z_j$,
  i.e., $(z_i,z_j)\in \compl{I} \circ \big((\progOrderRel\cup\atRel)\circ\compl{I}\big)^p$ for some $p$,
  there exist $a_1,b_1,\ldots,a_p,b_p\in\alphabet{G}$
  such that $(z_i,a_1)\notin I$,
  and for each $r=1,\ldots,p$,
  we have
  $\progOrder{a_r}{b_r}$ or $\at{a_r}{b_r}$,
  as well as $(b_r, a_{r+1})\notin I$ (where we set $a_{p+1}:=z_j$).

  For each $r=1,\ldots,p$,
  we define traces $\rho_r,\iota_r,\sigma_r$ as follows:
  \begin{itemize}
  \item If $\progOrder{a_r}{b_r}$ holds,
    then let $\rho_r\iota_r\sigma_r\in\traces{G}$ be the trace guaranteed by the definition of~$\progOrderRel$,
    where $\iota_r$ begins with $a_r$ and ends with $b_r$.
  \item If $\at{a_r}{b_r}$ holds,
    then let $\iota_r := \rho'b\iota'a\sigma'\in\traces{G_{\beta_k}}$ be the trace guaranteed by the definition of~$\atRel$.
    Let furthermore $\rho_r\, \rho'b\iota'a\sigma' \,\sigma_r \in \traces{G}$ (such a trace always exists).
  \end{itemize}
  Finally, let $\rho_0\, z_1\ldots z_m \,\sigma_0 \in \traces{G}$.
  Then, we define
  \begin{align*}
    \rho   &:= \ind{\rho_0}{1} \ind{\rho_1}{2} \ldots \ind{\rho_p}{p+1}\\
    \iota  &:= \ind{\iota_1}{2} \ldots \ind{\iota_p}{p+1}\\
    \sigma &:= \ind{\sigma_0}{1} \ind{\sigma_1}{2} \ldots \ind{\sigma_p}{p+1}
  \end{align*}
  For each $r=1,\ldots,p$,
  one of the following holds:
  \begin{itemize}
  \item We have $\progOrder{a_r}{b_r}$, and hence $\ind{a_r}{r+1}$ precedes $\ind{b_r}{r+1}$ in~$\iota$.
    Then, every covering interleaving must preserve the order between $\ind{a_r}{r+1}$, $\ind{b_r}{r+1}$ and $\ind{a_{r+1}}{(r+2) \bmod (p+1)}$,
    since $(b_r,a_{r+1})\notin I$.
  \item We have $\at{a_r}{b_r}$, and hence $\ind{b_r}{r+1}$ precedes $\ind{a_r}{r+1}$ in~$\iota$.
    Every covering interleaving preserves the order between $\ind{b_{r-1}}{r}$ (setting $b_0:=z_i$) and $\ind{a_r}{r+1}$,
    and between $\ind{b_r}{r+1}$ and $\ind{a_{r+1}}{(r+2)\bmod (p+1)}$.
    Hence, either the interleaving also preserves the order between $\ind{a_r}{r+1}$ and $\ind{a_{r+1}}{(r+2) \bmod (p+1)}$,
    or $\ind{a_{r+1}}{(r+2) \bmod (p+1)}$ occurs between $\ind{b_r}{r+1}$ and $\ind{a_r}{r+1}$.
    In the latter case, the atomic block containing $a_r,b_r$ occurs non-atomically in the interleaving.
  \end{itemize}

  We conclude that in all interleavings covering~$\tau$,
  either the atomic block $z_1\ldots z_m$ or some other atomic block occurs non-atomically.
  Hence, there does not exist an interleaving covering~$\tau$ in~$\liftIdx{\mu_\F}(\lang{\fuseblocks{P}{\F}})$,
  and the atomic fusion~$\F$ is unsound.

  \textbf{Soundness.}
  Assume that no $z_1\ldots z_m,i,j$ as described exist,
  and let $\tau$~be some interleaving in~$\lang{P}$.
  We show that $\tau$~can be transformed step-by-step,
  either by swapping commuting actions (up to the commutativity relation~$I'$ defined above) to derive a new interleaving covering~$\tau$,
  or by replacing a subsequence $\ind{z_1\ldots z_m}{t}$ with $z_1\ldots z_m\in\traces{G_{\beta_k}}$ with a special symbol $\ind{\beta_k^{z_1\ldots z_m}}{t}$,
  until there are no more non-atomic occurrences of atomic blocks in the final interleaving~$\tau^\star$.
  If we replace every~$\beta_k^\rho$ in~$\tau'$ by $\beta_k$, we derive an interleaving~$\hat{\tau}^\star\in\lang{\fuseblocks{P}{\F}}$.
  It follows that $\tau\in\liftIdx{\mu_\F}(\hat{\tau}^\star)$, and we are done.

  Towards this transformation,
  first replace every maximal atomically occurring subsequence $\langle z_1\ldots z_m:t \rangle$ with $z_1\ldots z_m\in\traces{G_{\beta_k}}$ (for some~$k$)
  by the corresponding~$\ind{\beta_k^{z_1\ldots z_m}}{t}$.
  The resulting interleaving~$\tau'$ is uniquely determined (every action can be uniquely associated with at most one~$G_{\beta_k}$, and $\traces{G_{\beta_k}}$ never contains an empty trace).
  It always holds that $\tau\in\liftIdx{\mu_\F}(\tau')$.
  If $\tau'\in\lang{\fuseblocks{P}{\F}}$, we are done.

  Otherwise, $\tau'$~must contain a subsequence $\sigma = \langle z_1\ldots z_i:t\rangle\ \iota\ \langle z_j\ldots z_m:t\rangle$,
  for some $z_1\ldots z_m\in\traces{G_{\beta_k}}$,
  where $\iota$~does not contain $\langle z_1:t \rangle$~nor~$\langle z_m:t\rangle$.
  I.e., the prefix $\langle z_1\ldots z_i:t\rangle$ and the suffix $\langle z_j\ldots z_m:t\rangle$ of~$\sigma$ belong to the same occurrence of the same atomic block.
  We have to show that every action~$\langle x:t'\rangle$ in~$\iota$ (where $t\neq t'$) can be moved to the left of~$\langle z_1:t\rangle$ resp.\ to the right of~$\langle z_m:t\rangle$.
  If we achieve this, we can once again replace the atomically occurring sequence $\langle z_1\ldots z_m:t\rangle$ by $\beta_k^{z_1\ldots z_m}$,
  and inductively conclude that any number of non-atomically occurring blocks in the original interleaving can be made atomic.

  Consider the first letter $\langle x:t'\rangle$ in $\iota$, i.e., let $\iota = \langle x:t'\rangle\ \iota'$.
  Wlog.\ we assume $t'\neq t$, as otherwise we can choose a larger $i$ and a shorter $\iota$.
  It must be the case that $(z_l,x)\in I'$ for all $l\leq i$, or $(x,z_r)\in I'$ for all $r > i$.
  Otherwise, we would immediately have
  $z_i \rightarrowtail z_j$:
  either because
  $(z_l,a)\notin I, \progOrder{a}{a}, (a,z_r)\notin I$ (in the case that $x=a\in\alphabet{G}$, noting that ${\progOrderRel}$ is reflexive),
  or (in the case that $x$ is some~$\beta_{k_x}^{\rho_x}$ with $\rho_x=a_1\ldots a_{m_x}$)
  $(z_l,a_i)\notin I$ and $(a_j,z_r)\notin I$, with $\progOrder{a_i}{a_j}$ if $j\geq i$ or $\at{a_i}{a_j}$ if $j<i$.
  This would contradict to our assumption.

  In the case that $(z_l,x)\in I'$ for all $l\leq i$, we see immediately that the sequence $\langle x:t'\rangle \langle z_1\ldots z_i:t\rangle\ \iota'\ \langle z_j\ldots z_m\rangle$ covers $\sigma$, and we are done.
  In the case that there exists $l\leq i$ with $(z_l,x)\notin I'$,
  we have $(x,z_r)\in I'$ for all $r > i$.
  We proceed by well-founded induction over the length of $\iota'$.
  For $\iota'= \varepsilon$, we are done immediately, as the sequence $\langle z_1\ldots z_i:t\rangle\ \langle z_j\ldots z_m\rangle\ \langle x:t'\rangle$ covers~$\sigma$ (and $j=i+1$).
  Otherwise, let $\iota' = \langle y:t'' \rangle\ \iota''$.

  If $t'\neq t''$ and $(x, y) \in I'$, then we swap $x$ and $y$, and continue to move~$x$ to the right of~$\iota''$,
  by applying the induction hypothesis to $\langle x:t'\rangle$ and $\iota''$,
  and to subsequently move~$x$ to the right of~$\ind{z_j\ldots z_n}{t}$.
  Otherwise, we have either $t'=t''$, or $t'\neq t''$ and $(x,y)\notin I'$.
  We distinguish four cases, and show in each that $(y,z_r)\in I'$ for all $r>i$.
  \begin{description}
  \item[Case 1: $x=a\in\alphabet{G},y=b\in\alphabet{G}$.]
    Either $t'=t''$, so we have $\progOrder{a}{b}$, or $(a,b)\notin I$.
    In either case, we must also have $(b,z_r)\in I\subseteq I'$ for all $r > i$,
    as otherwise we derive $z_l \rightarrowtail z_r$ (via $a$ and $b$),
    which contradicts our assumption.
  \item[Case 2: $x=\beta_{k_x}^{\rho_x}, y=b\in\alphabet{G}$.]
    Let $\rho_x=z_1^x \ldots z_{m_x}^x$, and let $i'$~be the index such that $(z_l,z_{i'}^x)\notin I$ (which exists, as $(z_l,x)\notin I'$).
    \begin{itemize}
    \item If $t'=t''$, we have $\progOrder{z_{i'}^x}{b}$.
      Hence, we must also have $(b,z_r)\in I\subseteq I'$ for all $r > i$:
      otherwise, we derive $z_l \rightarrowtail z_r$ (via $z_{i'}^x$ and $b$),
      which contradicts our assumption.
    \item Similarly, if $t'\neq t''$ and $(x,b)\notin I'$,
      taking~$j'$ as the index with $(z_{j'}^x,b)\notin I$,
      we have either $\progOrder{z_{i'}^x}{z_{j'}^x}$ (if $i'\leq j'$)
      or $\at{z_{i'}^x}{z_{j'}^x}$ (if $i'>j'$).
      In either case, we must also have $(b,z_r)\in I\subseteq I'$ for all $r > i$,
      to avoid deriving $z_l \rightarrowtail z_r$ (via $z_{i'}^x$ and $z_{j'}^x$), contradicting our assumption.
    \end{itemize}
  \item[Case 3: $x=a\in\alphabet{G}, y=\beta_{k_y}^{\rho_y}$.]
    Let $\rho_y=z_1^y\ldots z_{m_y}^y$.
    We show that for all $j'=1,\ldots,m_y$ and $r>i$,
    we have $(z_{j'}^y, z_r)\in I$.
    \begin{itemize}
    \item If $t'=t''$, we have $\progOrder{a}{z_{j'}^y}$.
      Then we must have $(z_{j'}^y,z_r)\in I$, otherwise we derive $z_l \rightarrowtail z_r$ (via $a$ and $z_{j'}^y$).
    \item If $t'\neq t''$ and $(a,y)\notin I'$, let $j''$ be the index with $(a,z_{j''}^y)\notin I$.
      We have either $\progOrder{z_{j''}^y}{z_{j'}^y}$ (if $j''\leq j'$) or $\at{z_{j''}^y}{z_{j'}^y}$ (if $j''<j'$).
      In both cases, we must have $(z_{j'}^y, z_r)\in I$,
      otherwise we can derive $z_l \rightarrowtail z_r$ (via $a$, $z_{j''}^y$ and $z_{j'}^y$).
    \end{itemize}
    It follows that $(y,z_r)\in I'$ for all $r>i$.
  \item[Case 4: $x=\beta_{k_x}^{\rho_x}, y=\beta_{k_y}^{\rho_y}$.]
    Let $\rho_x=z_1^x \ldots z_{m_x}^x$ and $\rho_y=z_1^y\ldots z_{m_y}^y$.
    Let $i'$~be the index such that $(z_l,z_{i'}^x)\notin I$ (which exists, as $(z_l,x)\notin I'$).
    \begin{itemize}
    \item If $t'=t''$, we have $\progOrder{z_{i'}^x}{z_{j'}^y}$ for all $j'=1,\ldots,m_y$.
      As before, to avoid deriving $z_l \rightarrowtail z_r$ (via $z_{i'}^x$ and $z_{j'}^y$),
      we must have $(z_{j'}^y,z_r)\in I$ for all $r>i$,
      and so $(y,z_r)\in I'$ for all $r>i$.
    \item If $t'\neq t''$ and $(x,y)\notin I'$, let $j_x,j_y$ be the indices such that $(z_{j_x}^x, z_{j_y}^y)\notin I$.
      We know that either $\progOrder{z_{i'}^x}{z_{i_x}^x}$ (if $i' \leq i_x$)
      or $\at{z_{i'}^x}{z_{i_x}^x}$ (if $i'>i_x$).
      For all $j'\geq j_y$, we have $\progOrder{z_{j_y}^y}{z_{j'}^y}$;
      whereas for all $j'< j_y$, we have $\at{z_{j_y}^y}{z_{j'}^y}$.
      In either case, we know that $(z_{j'}^y, z_r)\in I$ for all $r>i$,
      as we would otherwise derive $z_l\rightarrowtail z_r$ (via $z_{i'}^x$, $z_{i_x}^x$, $z_{i_y}^y$ and $z_{j'}^y$).
      Once again, it follows that $(y,z_r)\in I'$ for all $r>i$.
    \end{itemize}
  \end{description}
  Thus, we can apply our induction hypothesis to $\langle y:t'' \rangle$ and $\iota''$,
  allowing us to move $y$ to the right of~$\ind{z_m}{t}$,
  and yielding a covering trace with a segment $\langle z_1\ldots z_i:t\rangle\,\langle x:t'\rangle\, \iota''' \,\langle z_j\ldots z_m:t\rangle$.
  We apply the induction hypothesis again, this time to $\langle x:t'\rangle$ and $\iota'''$.
  This gives us a covering trace with a segment $\langle z_1\ldots z_i:t\rangle\,\iota'''' \,\langle z_j\ldots z_k:t\rangle$,
  in which we have successfully moved $\ind{x}{t'}$ out of the atomic block.

  This process is repeated, until $\ind{z_1\ldots z_m}{t}$ appears atomically,
  at which point we replace it by~$\beta_k^{z_1\ldots z_m}$.
  By induction over the number of atomic block occurrences that have not yet been replaced by some~$\beta_{k'}^{\rho}$,
  it follows that the transformation eventually yields an interleaving~$\tau^\star$ as described above.
  \qed
\end{proof}
Using \cref{conj:polynomial-sound-atomic},
the following algorithm decides soundness of~$\F$.
As the bodies~$G_{\beta_k}$ of atomic blocks may contain complex control flow including loops,
we cannot simply enumerate all possible~$z_1\ldots z_m\in\traces{G_{\beta_k}}$.
Hence, the algorithm reasons on the level of strongly connected components (SCCs) of~$G_{\beta_k}$.%
\begin{thmalgorithm}
\label{algo:check-block-postcond}
Execute the following steps:
\begin{enumerate}[label={\upshape Step~\arabic*:}]
\item Compute the graph $(\alphabet{G}, {\rightarrowtail})$.
  Then perform the subsequent steps for each~$\beta_k$.

\item Compute the SCCs of the edges of $G_{\beta_k}$,
  along with the reachability relation~${\preceq}$ between SCCs (a partial order).

\item Compute $\min(a)$ for each~$a\in\alphabet{G}$,
  the set of ${\preceq}$-minimal SCCs of $G_{\beta_k}$ containing some edge labeled by an action~$b\in\alphabet{G_{\beta_k}}$
  with $(b,a)\notin I$.

\item Compute $\max(a)$ for each~$a\in\alphabet{G}$,
  the set of ${\preceq}$-maximal SCCs of $G_{\beta_k}$ containing some edge labeled by an action~$b\in\alphabet{G_{\beta_k}}$
  with $(a,b)\notin I$.

\item Check if there exist SCCs $S_1,S_2$ of $G_{\beta_k}$
  such that
  \begin{itemize}
  \item $S_1 \preceq S_2$, and if $S_1=S_2$ then $S_1$ is a non-trivial SCC%
    \footnote{A non-trivial SCC must either contain at least two elements, or its single element must have a self-loop.},
    and
  \item the graph $(\alphabet{G},{\rightarrowtail})$ has an edge from some~$a$ to some~$b$ with $S_1\in\min(a)$ and $S_2\in\max(b)$.
  \end{itemize}
\end{enumerate}
If such $S_1,S_2$ exist (for some~$\beta_k$), then $\F$~is unsound. Otherwise, $\F$~is sound.
\end{thmalgorithm}
\begin{toappendix}
\begin{lemma}
  \label{thm:soundness-check-block-postcond}
  If \cref{algo:check-block-postcond} claims that $\F$ is sound,
  then $\F$ is indeed sound.
\end{lemma}
\begin{proof}
  We show the contraposition.
  Hence, assume that $\liftIdx{\mu_\F}(\lang{\fuseblocks{P}{\F}})$ is \emph{not} a Mazurkiewicz reduction of $\lang{P}$.
  By \cref{conj:polynomial-sound-atomic}, we know that there exists a word $z_1\ldots z_m\in \traces{G_{\beta_k}}$ for some~$k$,
  and indices $i<j$,
  such that $z_i \rightarrowtail z_j$.
  In particular, let $(z_i,a)\notin I$,
  $(b,z_j)\notin I$, with $(a,b)\in (\progOrderRel\cup\atRel)\comp (\compl{I}\comp (\progOrderRel\cup\atRel)^*)$.

  Consider the path through~$G_{\beta_k}$ for $z_1\ldots z_m$,
  and let $S_1$ be the SCC of the $i$-th edge resp.\ $S_2$ the SCC of the $j$-th edge in this run.
  Then we know that $S_1'\in\min(a)$ and $S_2'\in\max(b)$,
  for some $S_1'\preceq S_1$ and $S_2\preceq S_2'$.
  As the $j$-th edge is clearly reachable from the $i$-th edge,
  we know that $S_1 \prec S_2$ or $S_1=S_2$ is non-trivial (due to $i<j$, it cannot be trivial).
  Furthermore, by definition, the graph~$(\alphabet{G},{\rightarrowtail})$ has an edge from~$a$ to~$b$.

  Thus, \cref{algo:check-block-postcond} does not conclude that $\F$ is sound,
  and instead declares it unsound.
  \qed
\end{proof}

\begin{lemma}
  \label{thm:completeness-check-block-postcond}
  If \cref{algo:check-block-postcond} claims that $\F$ is unsound,
  then $\F$ is indeed unsound.
\end{lemma}
\begin{proof}
  Let $S_1,S_2$ be two SCCs of $G_{\beta_k}$, for some~$k$, such that either $S_1 \prec S_2$ or $S_1=S_2$ is nontrivial,
  and let $a,b\in\alphabet{G}$ such that $a\rightarrowtail b$, $S_1\in\min(a)$ and $S_2\in\max(b)$.
  Using \cref{conj:polynomial-sound-atomic}, we only have to show that there exists a word $z_1\ldots z_m\in\traces{G_{\beta_k}}$ and indices $i<j$
  with $z_i \rightarrowtail z_j$.

  As $S_1\in\min(a)$, there must exist an edge in $S_1$ labeled by some $z$ such that $(z,a)\notin I$.
  Similarly, as $S_2\in\max(b)$, there must exist an edge in $S_2$ labeled by some $z'$ with $(b,z')\notin I$.
  As $S_1 \preceq S_2$, there exists a path in~$G_{\beta_k}$ from some (and thus every) edge in $S_1$ to some (and thus every) edge in $S_2$.
  In particular, there is a path in~$G_{\beta_k}$ from the edge labeled by $z$ to the edge labeled by $z'$.
  Both in the case that $S_1 \prec S_2$ as well as in the case that $S_1=S_2$ is nontrivial,
  we can assume the path from $z$ to $z'$ to contain at least two edges
  (if we use the same edge for $z$ and $z'$, it must be in a nontrivial SCC and we can thus unroll some loop once to ensure this).
  By our assumption that all states of~$G_{\beta_k}$ are reachable and can reach an accepting state,
  this path can be embedded in a word $z_1\ldots z_m\in \traces{G_{\beta_k}}$ with $z_i=z$ and $z_j=z'$ for some $1 \leq i < j \leq m$.
  Note that $i<j$ relies particularly on the above assumption that the path from $z$ to $z'$ contains at least two edges.
  We then have $z_i\rightarrowtail z_j$, via~$a$ and~$b$.
  \qed
\end{proof}
\end{toappendix}

\begin{theoremrep}
  \Cref{algo:check-block-postcond} decides soundness of atomic fusions, and terminates in polynomial time.
\end{theoremrep}
\begin{proof}
  It follows from \cref{conj:polynomial-sound-atomic,thm:soundness-check-block-postcond,thm:completeness-check-block-postcond} that \cref{algo:check-block-postcond} always returns the correct result.

  To see that \cref{algo:check-block-postcond} terminates in polynomial time,
  consider that each step can be implemented in polynomial time:
  \begin{enumerate}[label={Step \arabic*:}]
  \item $\progOrderRel$ and $\atRel$ can be decided in linear time, $I$ is given.
    The transitive closure between all pairs is possible in cubic time.
  \item SCCs can be computed in time polynomial in the size of~$G_{\beta_k}$.
  \item For each~$a\in\alphabet{G}$, $\min(a)$ can be computed in time $O(|G|)$.
  \item Similarly for $\max(a)$.
  \item This step can also clearly be done in polynomial time.
  \end{enumerate}

  Hence, the theorem holds.
  \qed
\end{proof}

\subsection{Deciding Soundness of Synchronous Reductions}
\label{sec:check-barriers}
Next, we consider the soundness problem for synchronous reductions.
Hence, let $P = ((\emptyset, \syncPred[\top]), G)$ once again be a program over the trivial synchronization alphabet,
and let $G'$~be a sync-point instrumentation of~$P$.
We are interested in deciding whether $G'$~is sound,
i.e., whether $\lang{\addbarriers{P}{G'}}$ is a Mazurkiewicz reduction of $\lang{P}$, up to the commutativity relation~$I$.
Recall that $\lang{\addbarriers{P}{G'}}$ is given by
\begin{multline}
  \lang{\addbarriers{P}{G'}} = \{\, \tau \in \Actions_\indexed^* \mid \exists \hat{\tau}\in(\Actions\cup\S_\barrier)_\indexed^* \,.\, \tau = \hat{\tau}|_{\Actions_\indexed}
    \text{ and } \sync[\barrier]{\hat{\tau}}\\
    \text{ and } \forall i\in\N\,.\,
      \localTrace{\hat{\tau}}{i} \in \traces{G'} \cup \{\varepsilon\}
     \,\}
\end{multline}
By the definition of~$\syncPred[\barrier]$ (\cref{eq:sync-barrier}),
the underlying traces~$\hat{\tau}$ with~$\sync[\barrier]{\hat{\tau}}$ can be split into ``phases''
consisting of actions from~$\Actions$, separated by occurrences of sequences~$\langle \barrier:t_1\rangle\ldots \langle\barrier:t_n\rangle$, where $t_1,\ldots,t_n$~are the threads that are still running.
In order for the sync-point reduction to be sound, it must be possible to bring every interleaving of the original program $P$ into such a ``phase form'' by swapping commuting actions
and then inserting occurrences of~$\langle \barrier:t_1\rangle\ldots \langle\barrier:t_n\rangle$ in the right places.
Based on this idea, we define a ``phase order'' on actions:
\begin{definition}
  The \emph{phase order} is the relation $\phaseOrder \subseteq \alphabet{G}^2$,%
  with $(a, b) \in \phaseOrder$ if there exists a trace $\hat{\tau}\in(\Actions\cup\S_\barrier)_\indexed^*$
  with $\sync[\barrier]{\hat{\tau}}$ and $\hat{\tau}|_{\Actions_\indexed}\in\lang{\addbarriers{P}{G'}}$
  such that for some $i\neq j$,
  we have $\localTrace{\hat{\tau}}{i}=\tau_1 a \sigma_1$, $\localTrace{\hat{\tau}}{j}=\tau_2 b \sigma_2$,
  and $|\tau_1|_{\bullet} < |\tau_2|_{\bullet}$.
\end{definition}
Intuitively, the phase order captures that $a$~can occur in a strictly later phase than~$b$,
in some execution of the instrumented program.
As we are considering programs without synchronization (other than the newly-introduced sync-point~$\bullet$), we can break this condition down further:
\begin{observationrep}
  \label{obs:pho-simple-condition}
  We have $(a, b)\in\phaseOrder$ if and only if there exist two traces $\rho_1 a \sigma_1, \rho_2 b \sigma_2 \in \traces{G'}$
  with $|\rho_1|_\bullet < |\rho_2|_\bullet$.
\end{observationrep}
\begin{proof}
  Let $(a,b)\in\phaseOrder$.
  Let $\hat{\tau},\rho_1,\sigma_1,\rho_2,\sigma_2$ be as in the definition of~$\phaseOrder$.
  Then by definition,
  we have $\rho_1a \sigma_1, \rho_2b \sigma_2 \in \traces{G'}$, and we are done.

  Let now $\rho_1,\rho_2,\sigma_1,\sigma_2\in\alphabet{G'}$
  such that $\rho_1 a \sigma_1, \rho_2 b \sigma_2 \in \traces{G'}$,
  and $|\rho_1|_\barrier < |\rho_2|_\barrier$ holds.
  Let $\tau_{1,1},\ldots,\tau_{1,k_1} \in \Actions^*$ be the trace fragments such that $\rho_1 a \sigma_1 = \tau_{1,1}\barrier \ldots \barrier \tau_{1,k_1}$.
  Similarly for $\rho_2 b \sigma_2$ and $\tau_{2,1},\ldots,\tau_{2,k_2}$.
  Wlog.\ assume $k_1 \leq k_2$.
  Then the combined trace
  \[
    \tau = \langle \tau_{1,1}:1\rangle\,\langle \tau_{2,1}:2\rangle \bullet \ldots \bullet \langle\tau_{1,k_1}:1\rangle\,\langle\tau_{2,k_1}:2\rangle \bullet \ind{\tau_{2,k_1+1}}{2} \barrier \ldots \barrier \ind{\tau_{2,k_2}}{2}
  \]
  is a trace of $\lang{\addbarriers{P}{G'}}$,
  with a prefix $\langle \rho_1 a : 1\rangle$ of $\localTrace{\tau}{1}$
  and a prefix $\langle \rho_2 b : 2 \rangle$ of $\localTrace{\tau}{2}$,
  and hence we have $(a,b)\in\phaseOrder$.
  \qed
\end{proof}
Given~$\phaseOrder$, it is straightforward to decide soundness of the reduction:
\begin{propositionrep}
  \label{conj:barrier-sound-pho}
  The sync-point instrumentation~$G'$ is sound
  if and only if
  $\phaseOrder \subseteq I^{-1}$ holds.
\end{propositionrep}
\begin{proofsketch}
  If there exists a pair $(a,b)\in\phaseOrder \setminus I^{-1}$, we construct an interleaving (containing $a$~and~$b$),
  and use the injectivity condition in~\cref{def:barrier-instr} to show that no representative for this interleaving exists in~$\lang{\fuseblocks{P}{G}}$.

  Conversely, if $\phaseOrder \subseteq I^{-1}$ holds, we show that every $\tau\in\lang{P}$ does have a representative in~$\lang{\fuseblocks{P}{G}}$.
\end{proofsketch}
\begin{proof}
  First, assume there are actions with $(a, b)\in\phaseOrder$ but $(b, a) \notin I$.
  Let $\hat{\tau}$, $\tau_1,\sigma_1,\tau_2,\sigma_2$ be as in the definition of $\phaseOrder$.
  The interleaving $\tau:=\langle \tau_2 |_\Actions \,b : 2\rangle\, \langle \tau_1|_\Actions\,a : 1\rangle\, \langle \sigma_2|_\Actions : 2 \rangle\,\langle \sigma_1|_\Actions:1\rangle \in \lang{P}$
  has no representative in $\lang{\addbarriers{P}{G'}}$.
  To see this, suppose there existed such a representative $\hat{\tau}'|_{\Actions_\indexed}$,
  with $\sync[\barrier]{\hat{\tau}'}$,
  $\localTrace{\hat{\tau}'}{1}|_{\Actions} = (\tau_1a\sigma_1)|_{\Actions}$,
  $\localTrace{\hat{\tau}'}{2}|_\Actions = (\tau_2 b\sigma_2)|_\Actions$,
  and $\tau \coveredby{I} \hat{\tau}'|_{\Actions_\indexed}$.
  By the injectivity of projection to~$\Actions$, as required for sync-point instrumentations,
  we can conclude that $\localTrace{\hat{\tau}'}{1} = \tau_1\,a\,\sigma_1$ and $\localTrace{\hat{\tau}'}{2} = \tau_2\,b\,\sigma_2$.
  Hence, the prefix of $\hat{\tau}'$ up to the relevant occurrence of $\langle a:1\rangle$ contains $|\tau_1|_\bullet$ occurrences of $\bullet$,
  and the prefix of $\tau'$ up to the relevant occurrence of $\langle b:2\rangle$ contains $|\tau_2|_\bullet$ occurrences of $\bullet$.
  As $\tau \coveredby{I} \hat{\tau}'|_{\Actions_\indexed}$ and $(b,a)\notin I$, we have that the relevant occurrence of $\ind{b}{2}$ in $\hat{\tau}'$ (as in $\tau$)
  appears before the relevant occurrence of $\langle a:1\rangle$.
  Since we have~$\sync[\barrier]{\hat{\tau}'}$, this contradicts the fact that $|\tau_1|_\bullet < |\tau_2|_\bullet$.
  Thus, our supposition of the existence of $\hat{\tau}'$ was incorrect,
  and the sync-point instrumentation~$G'$ is unsound.
  \medskip

  Second, assume that $\phaseOrder \subseteq I^{-1}$ holds,
  and take an arbitrary trace $\tau\in\lang{P}$.
  Since we have $\traces{G'}|_\Actions = \traces{G}$,
  there exists a trace $\hat{\tau}\in(\Actions\cup\S_\barrier)^*$
  with $\localTrace{\hat{\tau}}{i}\in\{\varepsilon\}\cup\traces{G'}$ for every~$i$,
  and $\hat{\tau}|_{\Actions_\indexed} = \tau$.
  Note that we do not yet demand that~$\sync[\barrier]{\hat{\tau}}$ holds.
  Consider an extended commutativity relation~$I' \subseteq (\alphabet{G}\cup\S_\barrier)^2$,
  in which $\barrier$ commutes against everything (in either direction).
  By the assumption $\phaseOrder \subseteq I^{-1}$,%
  there exists some~$\hat{\tau}'$ covering $\hat{\tau}$ (up to~$I'$) such that $\sync[\barrier]{\hat{\tau}'}$ holds.
  We conclude that $\hat{\tau}'|_{\Actions_\indexed}\in\lang{\addbarriers{P}{G'}}$ is a representative of~$\tau$.
  \qed
\end{proof}

The fact that the above proof uses the injectivity condition of~\cref{def:barrier-instr} only for one direction of the equivalence
implies that even for an extended class of sync-point instrumentations that do not satisfy injectivity,
$\phaseOrder \subseteq I^{-1}$ implies soundness of the instrumentation.
Our approach thus remains sound in this extended class, though not complete.

To check soundness of a sync-point instrumentation,
the only remaining difficulty is to compute $\phaseOrder$.
In a program model without synchronization, this is possible in polynomial time, using the following approach:

\begin{thmalgorithm}
\label{algo:compute-pho}
  We compute $\phaseOrder$ as follows:

  \begin{enumerate}[label={\upshape Step~\arabic*:}]
  \item Compute $\min_\bullet(a) := \min\{\, |\tau|_\bullet \mid \exists \sigma\,.\,\tau a\sigma \in \traces{G'} \,\}$ for each action~$a$,
    by a shortest-path computation on~$G'$,
    where the initial location of~$G'$ is the source,
    locations enabling~$a$ are the targets,
    and the weight of an edge is 1 if labeled by $\barrier$ and 0 otherwise.
  \item Compute $\max_\bullet(a) := \max\{\, |\tau|_\bullet \mid \exists \sigma\,.\, \tau a \sigma \in \traces{G'} \,\} \in \N \cup \{\infty\}$ for each action~$a$.

    To do so, consider all loop-free paths from the initial location to a location enabling $a$.
    Determine if any such path can be pumped with a loop that contains~$\bullet$ (the loop may also include~$a$ itself).
    If so, then $\max_\bullet(a)$ is~$\infty$.
    Otherwise, $\max_\bullet(a)$ is the maximum weight of the loop-free paths, where the weight of an edge is 1 if labeled by $\bullet$ and 0 otherwise.
  \item Check for all pairs $(a,b)$ whether $\min_\bullet(a) < \max_\bullet(b)$.
    If so, then $(a,b)\in\phaseOrder$ holds.
    Otherwise, we have $(a,b)\notin\phaseOrder$.
  \end{enumerate}
\end{thmalgorithm}

\begin{toappendix}
  \begin{lemma}
    \label{lem:compute-pho-polynomial}
    \Cref{algo:compute-pho} computes $\phaseOrder$ in polynomial time.
  \end{lemma}
  \begin{proof}
    The algorithm decides the condition of~\cref{obs:pho-simple-condition}.
    Step~1 can be implemented via Dijkstra's algorithm in linear time.
    Step~2 can be implemented by a nested search (first identifying all locations on a path to~$a$, and then for each such location checking if it can reach itself again) in quadratic time.
    Finally, Step~3 is a simple constant-time check for each of the quadratically many pairs~$(a,b)$.
    \qed
  \end{proof}
\end{toappendix}

\begin{theoremrep}
  Soundness of sync-point instrumentations (in programs with the trivial synchronization alphabet)
  can be decided in polynomial time.
\end{theoremrep}
\begin{proof}
  It suffices to run~\cref{algo:compute-pho}, which requires polynomial time as per~\cref{lem:compute-pho-polynomial},
  and then checking the inclusion given in~\cref{conj:barrier-sound-pho}, which requires also only quadratic time.
  \qed
\end{proof}

\subsection{Deciding Soundness of Natural Reductions}
\label{sec:check-natural-reductions}
Let $P$~be a program, and $(\F,G')$~be a natural reduction specification for~$P$.
When deciding whether $(\F,G')$ is sound,
it is of course desirable to make use of the algorithms introduced in the previous subsections.
And indeed, the following observation gives us some hope that soundness checking of complicated reductions
can be decomposed into smaller steps.
\begin{propositionrep}
  Let $L_1,L_2,L_3\subseteq\Actions_\indexed^*$ such that $L_1 \subseteq L_2 \subseteq L_3$.
  Then $L_1$~is a Mazurkiewicz reduction of~$L_3$
  if and only if
  $L_1$~is a Mazurkiewicz reduction of~$L_2$, and $L_2$~is a Mazurkiewicz reduction of~$L_3$.
\end{propositionrep}
\begin{proof}
\newcommand\cl[1]{\mathsf{cl}(#1)}
  Given any language~$L\subseteq \Actions_\indexed^*$,
  let its \emph{closure} $\cl{L}$ be the set of all $\tau\in\Actions_\indexed^*$ such that $\tau\coveredby{I}\tau'$ for some $\tau'\in L$.
  This closure satisfies the usual closure laws (extensivity, monotonicity, idempotence),
  and $L$ is a Mazurkiewicz reduction of~$L'$ if and only if $L\subseteq L' \subseteq \cl{L}$.

  Suppose $R_2$ is a reduction of $P$, i.e., $R_2 \subseteq P$ and $P \subseteq \cl{R_2}$.
  We know already by assumption that $R_2 \subseteq R_1$ and $R_1 \subseteq P$.
  Furthermore, we have
  \[
    P \subseteq \cl{R_2} \subseteq \cl{R_1}
  \]
  by monotonicity of the closure,
  and
  \[
    R_1 \subseteq P \subseteq \cl{R_2}
  \]
  by assumption.
  Hence, $R_1$ is a reduction of $P$ and $R_2$ is a reduction of $R_1$.

  For the opposite direction, assume that $P \subseteq \cl{R_1}$ and $R_1 \subseteq \cl{R_2}$.
  Then we have $R_2 \subseteq P$ by transitivity (from assumptions),
  and
  \[
    P \subseteq \cl{R_1} \subseteq \cl{\cl{R_2}} \subseteq \cl{R_2}
  \]
  by assumption, monotonicity and idempotence of the closure.
  Hence, $R_2$ is a reduction of $P$.
\end{proof}
Recall that $(\F,G')$~is sound (wrt.~$I$) if $\liftIdx{\mu_\F}(\lang{\addbarriers{\fuseblocks{P}{\F}}{G'}})$ is a Mazurkiewicz reduction (up to~$I$) of~$\lang{P}$.
As an instantiation of the above result,
we get that $(\F,G')$~is sound if and only if
\begin{enumerate}[label=(\arabic*)]
\item $\liftIdx{\mu_\F}(\lang{\addbarriers{\fuseblocks{P}{\F}}{G'}})$ is a Mazurkiewicz reduction up to~$I$ of~$\liftIdx{\mu_\F}(\lang{\fuseblocks{P}{\F}})$,
\item and $\liftIdx{\mu_\F}(\lang{\fuseblocks{P}{\F}})$ is a Mazurkiewicz reduction up to~$I$ of~$\lang{P}$.
\end{enumerate}
Condition~(2) directly corresponds to soundness of the atomic fusion~$\F$,
and can therefore be decided using~\cref{algo:check-block-postcond}.
The situation is more complex for condition~(1),
as (due to the presence of~$\liftIdx{\mu_\F}$) it does not correspond directly to soundness of~$G'$.
Conceptually, it is not even immediately clear what soundness of~$G'$ would mean:
$\lang{\fuseblocks{P}{\F}}$ and $\lang{\addbarriers{\fuseblocks{P}{\F}}{G'}}$ are not languages over~$\alphabet{G}$;
they contain the block symbols introduced by~$\F$ (in addition to a subset of actions from~$\alphabet{G}$).
However, we can define the following commutativity relation, where $\F=(G'',(\beta_1,G_{\beta_1}),\ldots,(\beta_n,G_{\beta_n}))$,
and $a,b$ range over~$\alphabet{G}$:
\begin{align*}
  \tilde{I} :=\null &I\cap\alphabet{G'}^2\\
                    &\null \cup \{\, (a,\beta_k) \mid \forall b\in\alphabet{G_{\beta_k}}\,.\, (a,b)\in I \,\}\\
                    &\null \cup \{\, (\beta_k,b) \mid \forall a\in\alphabet{G_{\beta_k}}\,.\, (a,b)\in I \,\}\\
                    &\null \cup \{\, (\beta_{k_1},\beta_{k_2}) \mid \forall a\in\alphabet{G_{\beta_{k_1}}}\,\,\forall b\in\alphabet{G_{\beta_{k_2}}}\,.\, (a,b)\in I \,\}
\end{align*}
For this commutativity relation, we show:
\begin{lemmarep}
  Condition~(1) holds
  if and only if
  $\lang{\addbarriers{\fuseblocks{P}{\F}}{G'}}$ is a Mazurkiewicz reduction up to~$\tilde{I}$ of~$\lang{\fuseblocks{P}{\F}}$.
\end{lemmarep}
\begin{proof}
  Suppose condition~(1) holds.
  We have to show that $\phaseOrder \subseteq \tilde{I}^{-1}$.
  Hence, let $(x,y)\in\phaseOrder$.
  Note that both $x$ and $y$ can either be ``normal actions'' from~$\alphabet{G}$,
  or block symbols~$\beta_{k_1},\beta_{k_2}$.

  If they are normal actions $x=a,y=b$, then consider an interleaving of~$P$ in which $\ind{a}{1}$ occurs after~$\ind{b}{2}$
  -- as there is no synchronization, such an interleaving exists.
  However, by condition~(1) we know that there exists a representative of this interleaving that respects the sync-point semantics,
  and thus must have $\ind{b}{2}$ occur before $\ind{a}{1}$.
  Hence, $(b,a)\in I \subseteq \tilde{I}$ follows.

  In the case that one or both are block symbols~$\beta_{k_1},\beta_{k_2}$,
  we perform the same reasoning for all $a\in\alphabet{G_{\beta_{k_1}}}$ resp.\ $b\in\alphabet{G_{\beta_{k_2}}}$.
  Then, by the definition of~$\tilde{I}$, it also follows that $(y,x)\in\tilde{I}$.
\medskip

  Suppose now that conversely,
  $\lang{\addbarriers{\fuseblocks{P}{\F}}{G'}}$ is a Mazurkiewicz reduction up to~$\tilde{I}$ of~$\lang{\fuseblocks{P}{\F}}$.
  We know that this implies $\phaseOrder\subseteq \tilde{I}^{-1}$.
  In the case that $x=a,y=b$ are normal actions,
  this directly gives us $(b,a)\in I$ whenever $(a,b)\in\phaseOrder$.
  For block symbols, the definition of~$\tilde{I}$ similarly gives us $(b,a)\in I$ if $(x,y)\in\phaseOrder$,
  for all $a\in\alphabet{G_{\beta_{k_1}}}$ resp.\ $b\in\alphabet{G_{\beta_{k_2}}}$.
  Hence, any interleaving in $\liftIdx{\mu_\F}(\lang{\fuseblocks{P}{\F}})$ can be transformed to a covering interleaving in~$\liftIdx{\mu_\F}(\lang{\addbarriers{\fuseblocks{P}{\F}}{G'}})$.
  \qed
\end{proof}
This lemma allows us to decide condition~(1) using the approach from~\cref{sec:check-barriers}.
We conclude:
\begin{theorem}
  \label{thm:nat-red-soundness-polynomial}
  The soundness of a natural reduction specification for a program without synchronization
  is decidable in polynomial time.
\end{theorem}
Observe that, in the case of programs with synchronization,
our hardness result for atomic blocks carries over to the more general problem of soundness of natural reduction specifications.
\begin{corollary}
  For every synchronization alphabet,
  the coverability problem
  is polynomial-time reducible to the soundness of natural reduction specifications.
\end{corollary}
\begin{proof}
  Follows directly from~\cref{thm:atomic-blocks-coverability}.
\end{proof}

\section{On the Incompleteness of Mover Reasoning}
\label{sec:mover-incomplete}
There is a wealth of literature on sound reductions given by atomic blocks.
This literature builds on the seminal work of Lipton~\cite{lipton75:movers}
and the idea of \emph{movers}.
We show that, while sound, mover reasoning is incomplete (i.e., fails to admit some sound atomic fusions).
Completeness requires moving from mostly local reasoning (as in Lipton's rule)
to a global view (as in our algorithm).
A similar observation also holds for sync-points.

We begin by introducing movers and Lipton's rule for atomic blocks,
reformulated in our formal setting.
As before, we assume a program~$P=((\emptyset,\syncPred[\top]),G)$ and a commutativity relation~$I\subseteq \alphabet{G}^2$.
\begin{definition}[\cite{lipton75:movers}]
  An action~$a\in\alphabet{G}$ is:
  \begin{itemize}
  \item a \emph{left-mover} if $(b,a)\in I$ for all $b\in\alphabet{G}$,
  \item a \emph{right-mover} if $(a,b)\in I$ for all $b\in\alphabet{G}$,
  \item a \emph{both-mover} if it is both a left-mover and a right-mover,
  \item and a \emph{non-mover} if it is neither a left-mover nor a right-mover.
  \end{itemize}
\end{definition}
Lipton's rule is the following result:
\begin{proposition}[\cite{lipton75:movers}]
  An atomic fusion~$(G',(\beta_1,G_{\beta_1}), \ldots, (\beta_n,G_{\beta_n}))$ is sound
  if every $\tau\in \traces{G_{\beta_k}}$ (for every~$k$) can be written as $\tau=\tau_r\, a\, \tau_l$,
  for a  sequence of right-movers~$\tau_r$, some action~$a$ and a sequence of left-movers~$\tau_l$.
\end{proposition}

While Lipton's rule is sufficient to ensure soundness, it is not a necessary condition,
as demonstrated by the following example.
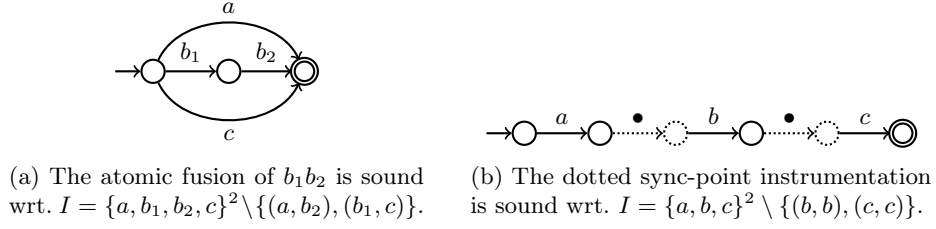
\begin{figure}[t]
\tikzstyle{st}=[draw,circle]
\begin{subfigure}{0.45\textwidth}
\centering
\begin{tikzpicture}[thick]
  \node[st] (l0) {};
  \node[st,right of=l0] (l1) {};
  \node[st,right of=l1,double] (l2) {};

  \draw[<-] (l0) -- ++(left:0.5);
  \draw[->] (l0) edge[bend left=65] node[auto]{$a$} (l2);
  \draw[->] (l0) -- node[auto]{$b_1$} (l1);
  \draw[->] (l1) -- node[auto]{$b_2$} (l2);
  \draw[->] (l0) edge[bend right=65] node[auto,swap]{$c$} (l2);
\end{tikzpicture}
\caption{The atomic fusion of $b_1b_2$ is sound wrt.\ $I=\{a,b_1,b_2,c\}^2\setminus\{(a,b_2),(b_1,c)\}$.}
\label{fig:lipton-incomplete}
\end{subfigure}
\hfill
\begin{subfigure}{0.5\textwidth}
\centering
\begin{tikzpicture}[thick]
  \node[st] (l0) {};
  \node[st,right of=l0]  (l1)  {};
  \node[st,right of=l1]  (l1') [densely dotted] {};
  \node[st,right of=l1'] (l2)  {};
  \node[st,right of=l2]  (l2') [densely dotted] {};
  \node[st,right of=l2',double] (l3) {};

  \draw[<-] (l0) -- ++(left:0.5);
  \draw[->] (l0) -- node[auto]{$a$} (l1);
  \draw[->,densely dotted] (l1) -- node[auto]{$\barrier$} (l1');
  \draw[->] (l1') -- node[auto]{$b$} (l2);
  \draw[->,densely dotted] (l2) -- node[auto]{$\barrier$} (l2');
  \draw[->] (l2') -- node[auto]{$c$} (l3);
\end{tikzpicture}
\caption{The dotted sync-point instrumentation is sound wrt.~$I=\{a,b,c\}^2\setminus\{(b,b),(c,c)\}$.}
\label{fig:barrier-movers-incomplete}
\end{subfigure}
\caption{Examples for incompleteness of mover reasoning.}
\end{figure}
\begin{example}
  Consider the thread template in~\cref{fig:lipton-incomplete},
  over $\Sigma=\{a,b_1,b_2,c\}$.
  Assume the given commutativity relation is  $I = \Sigma^2 \setminus \{(a, b_2), (b_1, c)\}$.

  Then, $b_1$ is a left-mover (but not a right-mover),
  whereas $b_2$ is a right-mover (but not a left-mover).
  Hence, Lipton's rule is not sufficient to conclude that $b_1 b_2$ can be fused in an atomic block.
  Yet, it \emph{is} sound to make $b_1b_2$ atomic.
  In any interleaving, all occurrences of~$a$ can be commuted to the beginning,
  and all occurrences of~$c$ can be commuted to the end of the interleaving.
  The remaining interleaved occurrences of~$b_1,b_2$ (from different threads)
  commute freely against each other,
  and can be reordered such that each~$\langle b_1 b_2:i\rangle$ executes atomically.

  Going further, if we consider $I'=I\setminus\{(b_1, b_2)\}$,
  the mover properties of all actions remain unchanged.
  Yet, the atomic fusion of $b_1b_2$ is unsound wrt.~$I'$:
  the interleaving $\langle b_1:1 \rangle\langle b_1:2 \rangle\langle b_2:1 \rangle\langle b_2:2 \rangle$ has no representative in which $\langle b_1b_2:1\rangle$ and $\langle b_1b_2:2\rangle$ execute atomically.
  Hence, mover properties are generally insufficient to characterize sound atomic fusions (even beyond Lipton's rule).
\end{example}
By contrast, our algorithm is sound and complete;
but it requires a global (rather than local) view that takes e.g.~control flow into account.

As a notable aside, in the case where $I$~is symmetric (every left- or right-mover is a both-mover),
Lipton's rule is indeed complete for parameterized programs.
\begin{propositionrep}
  Let $I$ be symmetric. If $\F=(G',(\beta_1,G_{\beta_1}), \ldots, (\beta_n,G_{\beta_n}))$ is sound wrt.~$I$,
  then all $\tau\in \traces{G_{\beta_k}}$ (for every~$k$) can be written as $\tau=\tau_r\, a\, \tau_l$,
  for a sequence of both-movers~$\tau_r$, some action~$a$, and a sequence of both-movers~$\tau_l$.
\end{propositionrep}
\begin{proofsketch}
  Assuming an atomic block that is not in Lipton form,
  we construct an interleaving of four threads (two of which execute the atomic block),
  such that no representative of this interleaving executes the block atomically.
  The construction makes use of the fact that the program is parameterized (in the number of threads) and that there is no synchronization between threads.
\end{proofsketch}
\begin{proof}
  We have to show that if there is an atomic block that is not in Lipton form,
  then there exists an interleaving without representative in~$\lang{\fuseblocks{P}{\F}}$.
  Hence, consider a program~$P$ and an atomic block that is not in Lipton form, i.e., it is of the form $\theta a_1 \eta a_2 \sigma$ with $a_1$ not a right-mover and $a_2$ not a left-mover.
  In other words, there exist actions~$b$ and~$c$ with $(a_1,b)\notin I$ and $(c,a_2)\notin I$.

  Consider now the interleaving $\tau = \langle a_1 \eta:1\rangle \langle b:3 \rangle \langle a_1\eta a_2:2 \rangle \langle c:4\rangle \langle a_2:1\rangle$.
  Such an interleaving can always be found in some trace of~$P$, by letting each thread run to the appropriate point first.
  In all traces $\tau'$ with $\tau \coveredby{I} \tau'$, we have that:
  \begin{itemize}
  \item $\langle a_1:1\rangle$ occurs before $\langle b:3 \rangle$,
    as $(a_1, b)\notin I$ by assumption.
  \item $\langle b:3\rangle$ occurs before $\langle a_1:2\rangle$, as \textbf{by symmetry}, $(b,a_1)\notin I$.
  \item $\langle a_1:2\rangle$ occurs before $\langle a_2:2\rangle$, by program order.
  \item $\langle a_2:2\rangle$ occurs before $\langle c:4\rangle$, as \textbf{by symmetry}, $(a_2, c)\notin I$.
  \item $\langle c:4\rangle$ occurs before $\langle a_2:1\rangle$, as $(c,a_2)\notin I$ by assumption.
  \end{itemize}
  Hence, no such $\tau'$ executes $\theta a_1 \eta a_2 \sigma$ atomically, and $\tau$ does not have a representative in $\lang{\fuseblocks{P}{\F}}$.
\end{proof}

We conclude this section by showing that mover reasoning is similarly unable to provide a characterization of sound sync-point instrumentations.
\begin{example}
  Consider the thread template and sync-point instrumentation shown in~\cref{fig:barrier-movers-incomplete},
  over actions $\Sigma=\{a,b,c\}$.
  This sync-point instrumentation is sound wrt.\ the commutativity relation  $I=\Sigma^2 \setminus \{(b,b),(c,c)\}$.
  However, this fact cannot be justified purely with mover properties.
  In this example, we have that $a$~is a both-mover, whereas $b$~and~$c$ are non-movers.
  If we consider a different commutativity relation~$I'=\Sigma^2\setminus\{(b,c),(c,b)\}$,
  the mover properties remain unchanged,
  yet the sync-point instrumentation is unsound wrt.~$I'$.
\end{example}
Unlike for atomic blocks, the inability of mover properties to characterize sound sync-point instrumentations thus persists even in the case of symmetric commutativity
(both $I$~and~$I'$ in the example above are symmetric).

\section{Reduction Soundness in the Presence of Locks}
\label{sec:soundness-with-locks}
Ignoring all synchronization between threads is a rather coarse abstraction.
One may be tempted to allow for at least some light-weight synchronization mechanisms to be considered during reduction checking.
For instance, locks (mutexes) are such a light-weight mechanism,
compared to more powerful synchronization via e.g.\ (global) boolean variables, broadcast, etc.
Locks are also ubiquitously used in concurrent programs.
However, we show in this section that allowing locks immediately leads to {\sc coNP}-hardness for the soundness problem of natural reductions, even in the simplest cases.

\paragraph{Atomic Blocks}
As shown in~\cref{thm:atomic-blocks-coverability},
coverability over programs with a certain synchronization alphabet
is polynomial-time reducible to the \emph{un}soundness of atomic fusions (with only a single atomic block) for the same class of programs.
For the case of synchronization via locks, we show:
\goodbreak
\begin{theoremrep}
  \label{thm:locks-cover-NP}
  The coverability problem over programs with locks is NP-hard.
\end{theoremrep}
\begin{proofsketch}
  We reduce 3-\textsc{Sat} to reachability in a bounded-thread program with locks, and then to coverability in a parameterized program with locks.
  For the latter step, we construct a thread template where each thread of the bounded-thread program is a separate branch, protected by a dedicated lock to ensure this branch can only be taken by one thread instance.
  For the reduction from 3-\textsc{Sat}, we introduce one lock~$m_i^l$ for each pair of a literal~$l$ over the propositional variables and a clause~$i$.
  There is one thread for each clause~$i$, which consists of choosing one of the literals~$l$ in the clause to satisfy,
  acquiring the lock~$m_i^l$, checking that no thread for some other clause~$j$ has made a contradictory choice (i.e., $m_j^{\lnot l}$~must be free), and reaching some location~$\ell_i$.
  The configuration of all~$\ell_i$ is reachable if and only if the conjunction of all clauses is satisfiable.
\end{proofsketch}
\begin{proof}
  
Let us first note that in a program that synchronizes only over locks,
coverability of a configuration is equivalent to reachability of the exact same configurations:
additional threads can only impede an execution (by acquiring locks) but not enable additional behaviours.

As described in the proof sketch, we reduce \textsc{3-Sat} to reachability in a bounded-thread program, and then reduce that reachability problem to coverability in the parameterized case.
We focus here on the reduction from \textsc{3-Sat}.

Given a \textsc{3-sat} formula $\varphi \equiv C_1 \land \ldots \land C_n$
with clauses $C_1,\ldots, C_n$ over propositional variables $\{x,y,\ldots\}$,
we construct a program with $n$ threads, one for each clause.
For each literal $l\in\{x,\lnot x,y,\lnot y, \ldots\}$,
we introduce $n$~locks $m^l_i$ with $i=1,\ldots,n$.
The idea is that $m^l_i$~being held (by thread~$i$, corresponding to~$C_i$) indicates that thread~$i$ has the "opinion" that $l$~evaluates to true.
The thread templates ensure that, in order to reach a certain configuration, all threads must have a consistent opinion, which corresponds to a (partial) satisfying assignment for~$\varphi$.

An example of such a thread template is shown in~\cref{fig:3sat-coverability-example}.
Generally, the thread template for a thread~$i$ consists of 3~branches, one for each literal~$l$ in~$C_i$.
In each branch, the thread first acquires the respective lock~$m^l_i$.
It then successively checks that all the locks~$m^{\lnot l}_j$ for $j\neq i$ are free, by acquiring and then releasing them.
Afterwards, it goes to a dedicated location $\ell_i$ (common end point of all branches).

\begin{figure}
\centering
\begin{tikzpicture}[thick]
  \node[draw,circle] (init) {};
  \node[draw,circle,below=1.5cm of init,xshift=-1.5cm] (x1) {};
  \node[draw,circle,below=1.5cm of init] (y1) {};
  \node[draw,circle,below=1.5cm of init,xshift=+1.5cm] (z1) {};

  \node[draw,circle,below of=x1] (x2) {};
  \node[draw,circle,below of=y1] (y2) {};
  \node[draw,circle,below of=z1] (z2) {};

  \node[draw,circle,below of=x2] (x3) {};
  \node[draw,circle,below of=y2] (y3) {};
  \node[draw,circle,below of=z2] (z3) {};

  \node[draw,circle,below of=x3] (x4) {};
  \node[draw,circle,below of=y3] (y4) {};
  \node[draw,circle,below of=z3] (z4) {};

  \node[draw,circle,below=1.5cm of y4,label={[font=\scriptsize]below:$\ell_i$}] (li) {};

  \draw[<-] (init) -- ++(up:7mm);
  \draw[->] (init) edge[bend right=45] node[auto,swap,font=\scriptsize]{$\mathit{acq}(m^x_2)$} (x1);
  \draw[->] (init) -- node[auto,font=\scriptsize]{$\mathit{acq}(m^{\lnot y}_2)$} (y1);
  \draw[->] (init) edge[bend left=45] node[auto,font=\scriptsize]{$\mathit{acq}(m^{\lnot z}_2)$} (z1);

  \draw[->] (x1) -- node[auto,font=\scriptsize]{$\mathit{acq}(m^{\lnot x}_1)$} (x2);
  \draw[->] (y1) -- node[auto,font=\scriptsize]{$\mathit{acq}(m^{y}_1)$} (y2);
  \draw[->] (z1) -- node[auto,font=\scriptsize]{$\mathit{acq}(m^{z}_1)$} (z2);

  \draw[->] (x2) -- node[auto,font=\scriptsize]{$\mathit{rel}(m^{\lnot x}_1)$} (x3);
  \draw[->] (y2) -- node[auto,font=\scriptsize]{$\mathit{rel}(m^{y}_1)$} (y3);
  \draw[->] (z2) -- node[auto,font=\scriptsize]{$\mathit{rel}(m^{z}_1)$} (z3);

  \draw[->] (x3) -- node[auto,font=\scriptsize]{$\mathit{acq}(m^{\lnot x}_3)$} (x4);
  \draw[->] (y3) -- node[auto,font=\scriptsize]{$\mathit{acq}(m^{y}_3)$} (y4);
  \draw[->] (z3) -- node[auto,font=\scriptsize]{$\mathit{acq}(m^{z}_3)$} (z4);

  \draw[->] (x4) edge[bend right=45] node[auto,swap,font=\scriptsize]{$\mathit{rel}(m^{\lnot x}_3)$} (li);
  \draw[->] (y4) -- node[auto,font=\scriptsize]{$\mathit{rel}(m^{y}_3)$} (li);
  \draw[->] (z4) edge[bend left=45] node[auto,font=\scriptsize]{$\mathit{rel}(m^{z}_3)$} (li);
\end{tikzpicture}
\caption{
  Thread template for a clause $C_2\equiv (x \lor \lnot y \lor \lnot z)$ in $\varphi \equiv C_1\land C_2\land C_3$.
}
\label{fig:3sat-coverability-example}
\end{figure}
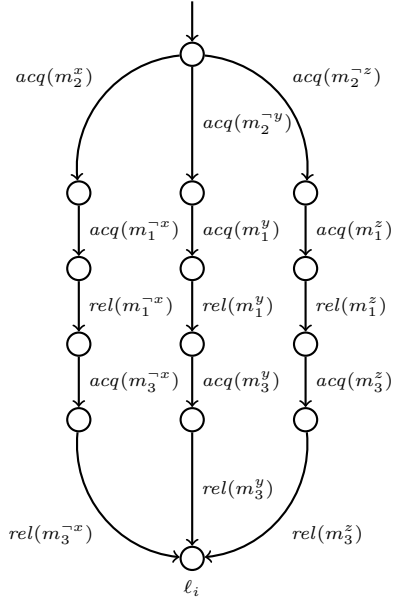

\begin{lemma}
  \label{lem:3sat-coverability-equiv}
  The configuration $[\ell_1,\ldots,\ell_n]$ is coverable (resp.\ reachable) in this program
  if and only if the \textsc{3-sat} formula~$\varphi$ is satisfiable.
\end{lemma}
\begin{proof}
  For the forwards direction, take a run that reaches the configuration $[\ell_1, \ldots, \ell_n ]$.
  In the final state, each thread~$i$ only holds some lock $m^x_i$ or $m^{\lnot x}_i$ (but never both at the same time).
  Moreover, it cannot be that $m^x_i$~and $m^{\lnot x}_j$ are both held (by threads~$i$ and~$j$, respectively):
  If wlog.\ thread~$i$ acquires $m^x_i$ before thread~$j$ acquires $m^{\lnot x}_j$,
  then thread~$j$ would block when trying to acquire $m^x_i$, and would not reach $\ell_j$.
  We define a satisfying partial assignment~$\rho$ for the CNF formula~$\varphi$ by setting $\rho(x)=\top$ if $m^x_i$ is held (by some thread~$i$),
  and $\rho(x)=\bot$ if $m^{\lnot x}_j$ is held (by some thread~$j$).
  By the considerations above, the partial assignment is well-defined.
  As each thread~$i$ holds the lock~$m^l_i$ for one of the literals~$l$ in $C_i$, this partial assignment satisfies each clause $C_i$.

  For the reverse direction, take a satisfying assignment~$\rho$ of the CNF formula~$\varphi$.
  We construct a run reaching $\langle \ell_1, \ldots, \ell_n \rangle$ by running the threads sequentially, one after the other.
  In particular, we let each thread as its first step chose some literal~$l$ that is satisfied by~$\rho$.
  Note that each thread~$j$ that has reached~$\ell_j$ only holds a single lock of the form $m^{l'}_j$.
  Hence, a thread~$i$ (with $i\neq j$) can always (after chosing a literal~$l$) execute its first lock acquisition of~$m^l_i$.
  Then, the thread~$i$ proceeds to test the availability of the locks $m^{\lnot l}_j$ corresponding to the negated literal (which is not satisfied by $\rho$).
  These are always available: if thread~$j$ ran before (and has already reached $\ell_j$), it may hold $m^l_j$ but not $m^{\lnot l}_j$ (by construction of our run);
  and if thread~$j$ did not run before, both locks~$m^l_j$ and~$m^{\lnot l}_j$ are available.
  Hence the constructed run is feasible wrt.\ the lock semantics, and it clearly reaches the configuration $\langle \ell_1, \ldots, \ell_n \rangle$.
  \qed (\cref{lem:3sat-coverability-equiv})
\end{proof}

As \textsc{3-sat} is well-known to be NP-complete,
we get that coverability over a program with locks is NP-hard.
\qed

\end{proof}

\begin{corollary}
  \label{corr:fusion-soundness-locks-hard}
  Atomic fusion soundness in programs with locks is {\sc coNP}-hard.%
\end{corollary}
\begin{proof}
  Follows directly from~\cref{thm:atomic-blocks-coverability} and \cref{thm:locks-cover-NP}.
\end{proof}

The corresponding upper-bound is (to our knowledge) an open question.

Furthermore, we note that previous work~\cite{farzan:predicting-atomicity-violations} has shown {\sc coNP}-hardness for a special case of atomic fusion soundness (\emph{conflict serializability})
in programs with a bounded number of threads (even without synchronization).
Interestingly, this previous result provides for an alternative proof of \cref{corr:fusion-soundness-locks-hard}:
We can reduce the soundness problem in parameterized programs to the bounded case, by using dedicated locks to limit the number of running threads.
Hence, while the shift from a bounded number of threads to parameterized programs significantly simplifies the soundness problem in the absence of synchronization,
lock synchronization yields back the hardness.
\begin{toappendix}
  Below we sketch the alternative proof of \cref{corr:fusion-soundness-locks-hard} discussed in \cref{sec:soundness-with-locks}.
\begin{proof}[Alternative proof of \cref{corr:fusion-soundness-locks-hard}.]
  We appeal to the {\sc NP}-complexity result for non-serializability of regular programs (a special case of atomic fusion unsoundness) over programs with a bounded number of threads
  shown in~\cite[Theorem 2]{farzan:predicting-atomicity-violations}.
  Similar to the proof of~\cref{thm:locks-cover-NP},
  we use the fact that we can transform a bounded-thread program with locks to a parameterized program with locks,
  by adding new lock variables $m_1,\ldots,m_k$ for the $k$ threads,
  and then combining the different thread templates as separate, completely disjoint branches in a single thread template.
  The branch for each thread~$i$ is protected by the lock~$m_i$, which is acquired initially but never released.
  Hence, there can be at most one thread instance executing the code of each thread from the bounded-thread program;
  any further threads of the parameterized program are stuck in their initial location and cannot execute any steps.
\end{proof}
\end{toappendix}

\paragraph{Synchronous reductions.}
For synchronous reductions (resp.\ sync-point instrumentations), we have not shown a general result relating coverability (for configurations of arbitrary-width) and reduction checking.
In fact, the algorithm in \cref{sec:check-barriers} at first suggests that only pairwise checks are needed (for the computation of the phase-order).
Yet, the difficulty lies in the fact that sync-points themselves are synchronization actions, and may interfere with other synchronization mechanisms such as locks.
Using this observation, we show:
\begin{theoremrep}
  The coverability problem of programs with locks is polynomial-time reducible to the sync-point soundness problem (of programs with locks).
\end{theoremrep}
\begin{proofsketch}
  Given a program~$P$ and a configuration~$C=[\ell_1, \ldots, \ell_n]$ for which to decide coverability,
  we modify the thread template as follows.
  We take fresh locks~$m_1,\ldots,m_n$, add a fresh location~$\ell_\exit'$, and chose~$\ell_\exit'$ as the new exit location.
  From each location~$\ell_i\in C$, we add a subgraph that first acquires~$m_i$, then branches over all~$j\neq i$ and checks that~$m_j$ is still free (by acquiring and releasing it),
  and then finally releases $m_i$, reaching~$\ell_\exit'$.
  For the sync-point instrumentation, we add a single sync-point~$\barrier$ immediately before the final step (releasing $m_i$).
  We note that~$C$ is coverable (in~$P$) if and only if it is reachable.
  Any interleaving~$\tau$ of~$P$ reaching a sub-configuration of~$C$ is an interleaving of the modified program (all threads can reach~$\ell_\exit'$).
  However, in an interleaving reaching exactly~$C$,
  the sync-point instrumentation forces at least one of the $n$~threads to deadlock when trying to acquire some~$m_j$ (and hence, all other threads are blocked at the sync-point).
  Thus, the sync-point instrumentation is sound (wrt.\ any~$I$) if and only if $C$~is not coverable.
\end{proofsketch}
\begin{proof}
  
Given a program $P$ with locks and a multiset of locations~$C=[\ell_1,\ldots,\ell_n]$ (where $n>1$),
take fresh locks $m_1,\ldots,m_n$ and a fresh exit location $\ell_\exit'$,
and extend the thread template by adding the following code as a branch from $\ell_i$:
\[
  \acq{m_i};\ \texttt{choose }j\in\{1,\ldots,n\}\texttt{ with } j\neq i \texttt{ {\Large\textbraceleft} } \acq{m_j};\ \rel{m_j} \texttt{ {\Large\textbraceright} };\ \rel{m_i}; \ \texttt{goto }\ell_\exit'
\]
In the control flow graph, the \texttt{choose} block is represented by $n-1$ branches, as seen in~\cref{fig:barriers-locks-reduction}.

\begin{figure}
\centering
\begin{tikzpicture}[node distance=1.75cm,thick]
  \node[draw,circle,label=165:$\ell_i$] (li) {};
  \node[draw,circle,label=165:$\ell_i^L$,below of=li]  (li1) {};
  \node[draw,circle,label=165:$\ell_i^{j,L}$,below left of=li1] (lij1) {};
  \node[draw,circle,label=195:$\ell_i^U$,below right of=lij1] (li3) {};
  \node[draw,double,circle,label=195:$\ell_\exit'$,below of=li3] (lex) {};
  \node[below right of=li1] (lin1) {\ldots};
  
  \draw[->] (li)  -- node[auto]{$\acq{m_i}$} (li1);
  \draw[->] (li1) edge[bend right=40] node[auto,swap]{$\acq{m_j}$} (lij1);
  \draw[->] (lij1) edge[bend right=40] node[auto,swap]{$\rel{m_j}$} (li3);
  \draw[->] (li1) edge[bend left=40] (lin1);
  \draw[->] (lin1) edge[bend left=40] (li3);
  \draw[->] (li3) -- node[auto]{$\rel{m_i}$} (lex);
\end{tikzpicture}
\caption{Thread template for the reduction of coverability with locks to soundness of sync-point instrumentation.}
\label{fig:barriers-locks-reduction}
\end{figure}
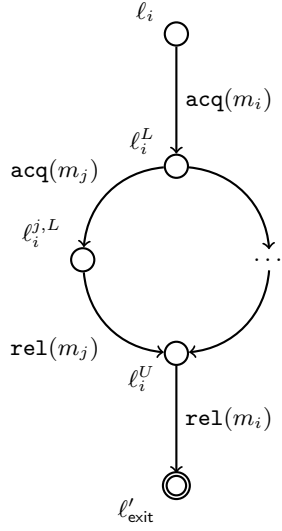

The idea is that if $C$ is coverable (i.e. reachable) with some interleaving~$\tau$,
then $\tau$ is contained in the language of the program $\hat{P}$ resulting from the above modification.
This is possible, because after reaching $C$, each thread $i$ can one-by-one go from $\ell_i$ to $\ell_\exit'$ completely (choosing any branch) without interruption; at the end all locks $m_i$ are free again (as they are after $\tau$, since they are fresh).
Note furthermore that in addition, $\hat{P}$ also includes other interleavings,
namely all interleavings $\tau'$ that reach a sub-configuration of $C$ (and only those interleavings).

Now consider the sync-point instrumentation where we add a single sync-point~$\bullet$ at each location $\ell_i^{U}$, just before $\rel{m_i}$.
An interleaving $\tau$ reaching exactly $C$ is no longer in the language of the instrumented program:
as all $n$ threads must pass the sync-point together, they wait at $\ell_i^U$ until all $n$ threads have reached this location.
However, the last thread that moves from $\ell_i$ to $\ell_i^L$ will be stuck there;
all its branches are blocked because (since it is the last thread), all $m_j$ are already held by the respective threads.
Consequently, the other threads forever wait at $\ell_i^U$ and cannot pass the sync-point.

On the other hand, all other interleavings $\tau'$ in the language of $\hat{P}$ that reach only a strict sub-configuration $C'$ of $C$ are still possible with the sync-point instrumentation.
If for instance $C' \subseteq_{ms} C - [\ell_j]$, then all threads~$i$ can execute $\acq{m_i};\acq{m_j}\rel{m_j}$ in turn (without interruption), pass the sync-point together, and release their respective $m_i$.

Now let $I$ be any commutativity.
Then the described sync-point instrumentation is sound (wrt.~$I$)
iff $C$ is not coverable in the original program:
If $C$ is coverable/reachable by $\tau$, the sync-point instrumentation loses the only possible representative for~$\tau$ (all remaining traces reach a strict sub-configuration of~$C$, hence there must be fewer than~$n$ active threads, so these traces cannot cover~$\tau$).
If $C$ is not coverable, the sync-point instrumentation does not change the language.
\qed

\end{proof}

\paragraph{Using partial information about locks.}
We conclude this section with a short outlook on how soundness checking for natural reductions can nevertheless benefit from (possibly incomplete, but sound) information about locks.
Suppose that we have information available about a set of locks~$M_\ell$ that a thread \emph{must} hold whenever it is at location~$\ell$.
Such information could for instance be derived from a static analysis,
or a thread-modular proof given to a deductive verifier by a user.
Then one can soundly declare that outgoing actions~$a_\ell$ of~$\ell$ and $a_{\ell'}$ of~$\ell'$ commute
whenever the sets $M_\ell$ and $M_{\ell'}$ overlap.

More generally, any kind of \emph{thread-modular invariants}~$\varphi_{\ell},\varphi_{\ell'}$
that hold whenever a thread is at location~$\ell$ resp.~$\ell'$
can be used to increase commutativity,
by determining if actions~$a_\ell,a_{\ell'}$ commute \emph{under the assumption}~$\varphi_{\ell}\land\varphi_{\ell'}$~\cite{popl24:parameterized}.
The idea for locks described above corresponds to the special case that~$\varphi_{\ell}\land\varphi_{\ell'}\equiv \bot$
(it is impossible that both threads hold the same lock),
under which the actions vacuously commute.

The resulting extended commutativity relation~$I' \supseteq I$
can increase the power of both mover reasoning~\cite{lipton75:movers} (cf.~\cref{sec:mover-incomplete})
as well as our algorithms.

\section{Related Work}

Checking conflict serializability for transactions can be seen as an instance of deciding soundness of atomic blocks, where the commutativity relation is fixed: read accesses commute and (read or write) accesses to different shared variables also commute. The complexity of checking serializability has been studied in \cite{bouajjani:concurrent-progs-seq-specs,farzan:predicting-atomicity-violations}, but these results do not extend to arbitrary commutativity relations as in our paper.

Techniques for proving soundness of synchronous reductions which are sound but possibly incomplete have been studied in~\cite{gleissenthal:pretend-synchrony,DBLP:conf/pldi/KraglEHMQ20}. These techniques are based on the mover types from Lipton's reduction theory. 

The works of Farzan et al.~\cite{pldi22:sound-seq,popl24:parameterized} define an algorithmic verification framework for concurrent programs that includes computing sound reductions. The syntactic representation of these reductions is somewhat complex, which is fine in the context of a fully automated, non-interactive proof. In our paper, we look at classes of reductions which are humanly readable and use simple syntactic constructs.

\section{Conclusion}
We have proposed \emph{natural reductions} of parameterized concurrent programs,
an approach for users to specify a subset of interleavings that may be easier to verify.
Natural reductions can be specified with ease, by adding atomic blocks and global rendez-vous points to the thread template.

We have studied the problem of deciding whether a natural reduction given by a user is indeed \emph{sound},
and can thus be soundly verified in place of the original program.
We link the complexity of this problem to the corresponding coverability problem.
As a consequence, the problem becomes intractable ({\sc coNP}-hard) as soon as even very light-weight synchronization between threads is allowed (in particular, locks).

This motivates us to take an abstract view of the program,
i.e., to abstract away from the semantics of locks and other synchronization operators.
We present the first complete decision procedure for checking soundness of natural reductions in this setting,
and show that it runs in polynomial time.
Hence, our approach overcomes the inherent incompleteness of previous work based on \emph{mover reasoning},
while retaining polynomial complexity.

In the future, our decision procedure could replace resp.\ complement mover reasoning in reduction-based deductive verifiers such as Civl~\cite{DBLP:conf/cav/HawblitzelPQT15} or Anchor~\cite{flanagan:anchor},
broadening the range of specifiable reductions and providing a guarantee of completeness.
As part of such an application,
a relevant question for study is the extension of our polynomial-time algorithm to other settings.
In particular, it is of interest to determine whether some form of synchronization (even more light-weight than locks), such as \emph{structured concurrent programs}, \emph{permissions}, or specific \emph{locking disciplines} can be accommodated.
Given that in such cases, it is typically easy to determine which parts of the program may run in parallel (syntactically resp.\ through a type system),
there is some hope that polynomial runtime can be retained.

\bibliographystyle{splncs04}
\bibliography{references}

\begin{thebibliography}{10}
\providecommand{\url}[1]{\texttt{#1}}
\providecommand{\urlprefix}{URL }
\providecommand{\doi}[1]{https://doi.org/#1}

\bibitem{bouajjani:concurrent-progs-seq-specs}
Bouajjani, A., Emmi, M., Enea, C., Hamza, J.: Verifying concurrent programs
  against sequential specifications. In: {ESOP}. pp. 290--309. Lecture Notes in
  Computer Science, Springer (2013). \doi{10.1007/978-3-642-37036-6\_17}

\bibitem{DBLP:conf/osdi/ChajedKLZ18}
Chajed, T., Kaashoek, M.F., Lampson, B.W., Zeldovich, N.: Verifying concurrent
  software using movers in {CSPEC}. In: {OSDI}. pp. 306--322. {USENIX}
  Association (2018)

\bibitem{DBLP:conf/podc/ChouG88}
Chou, C., Gafni, E.: Understanding and verifying distributed algorithms using
  stratified decomposition. In: {PODC}. pp. 44--65. {ACM} (1988).
  \doi{10.1145/62546.62556}

\bibitem{clerbout:semi-commutativity}
Clerbout, M., Latteux, M., Roos, Y.: Semi-commutations. In: The Book of Traces,
  pp. 487--552. World Scientific (1995). \doi{10.1142/9789814261456\_0012}

\bibitem{DBLP:conf/cav/DamianDMW19}
Damian, A., Dragoi, C., Militaru, A., Widder, J.: Communication-closed
  asynchronous protocols. In: {CAV} {(2)}. pp. 344--363. Lecture Notes in
  Computer Science, Springer (2019). \doi{10.1007/978-3-030-25543-5\_20}

\bibitem{DBLP:journals/scp/ElradF82}
Elrad, T., Francez, N.: Decomposition of distributed programs into
  communication-closed layers. Sci. Comput. Program.  \textbf{2}(3),  155--173
  (1982). \doi{10.1016/0167-6423(83)90013-8}

\bibitem{esparza:keeping-crowd-safe}
Esparza, J.: Keeping a crowd safe: On the complexity of parameterized
  verification (invited talk). In: {STACS}. pp. 1--10. LIPIcs, Schloss Dagstuhl
  - Leibniz-Zentrum f{\"{u}}r Informatik (2014).
  \doi{10.4230/LIPICS.STACS.2014.1}

\bibitem{pldi22:sound-seq}
Farzan, A., Klumpp, D., Podelski, A.: Sound sequentialization for concurrent
  program verification. In: {PLDI}. pp. 506--521. {ACM} (2022).
  \doi{10.1145/3519939.3523727}

\bibitem{popl24:parameterized}
Farzan, A., Klumpp, D., Podelski, A.: Commutativity simplifies proofs of
  parameterized programs. Proc. {ACM} Program. Lang.  \textbf{8}({POPL}),
  2485--2513 (2024). \doi{10.1145/3632925}

\bibitem{farzan:predicting-atomicity-violations}
Farzan, A., Madhusudan, P.: The complexity of predicting atomicity violations.
  In: {TACAS}. pp. 155--169. Lecture Notes in Computer Science, Springer
  (2009). \doi{10.1007/978-3-642-00768-2\_14}

\bibitem{farzan20:red-safety-proofs}
Farzan, A., Vandikas, A.: Reductions for safety proofs. Proc. {ACM} Program.
  Lang.  \textbf{4}({POPL}),  13:1--13:28 (2020). \doi{10.1145/3371081}

\bibitem{flanagan:anchor}
Flanagan, C., Freund, S.N.: The anchor verifier for blocking and non-blocking
  concurrent software. Proc. {ACM} Program. Lang.  \textbf{4}({OOPSLA}),
  156:1--156:29 (2020). \doi{10.1145/3428224}

\bibitem{gleissenthal:pretend-synchrony}
von Gleissenthall, K., Kici, R.G., Bakst, A., Stefan, D., Jhala, R.: Pretend
  synchrony: synchronous verification of asynchronous distributed programs.
  Proc. {ACM} Program. Lang.  \textbf{3}({POPL}),  59:1--59:30 (2019).
  \doi{10.1145/3290372}

\bibitem{DBLP:conf/cav/HawblitzelPQT15}
Hawblitzel, C., Petrank, E., Qadeer, S., Tasiran, S.: Automated and modular
  refinement reasoning for concurrent programs. In: {CAV} {(2)}. pp. 449--465.
  Lecture Notes in Computer Science, Springer (2015).
  \doi{10.1007/978-3-319-21668-3\_26}

\bibitem{DBLP:conf/pldi/KraglEHMQ20}
Kragl, B., Enea, C., Henzinger, T.A., Mutluergil, S.O., Qadeer, S.: Inductive
  sequentialization of asynchronous programs. In: {PLDI}. pp. 227--242. {ACM}
  (2020). \doi{10.1145/3385412.3385980}

\bibitem{lipton75:movers}
Lipton, R.J.: Reduction: {A} method of proving properties of parallel programs.
  Commun. {ACM}  \textbf{18}(12),  717--721 (1975). \doi{10.1145/361227.361234}

\bibitem{DBLP:conf/pldi/LorchCKPQSWZ20}
Lorch, J.R., Chen, Y., Kapritsos, M., Parno, B., Qadeer, S., Sharma, U.,
  Wilcox, J.R., Zhao, X.: Armada: low-effort verification of high-performance
  concurrent programs. In: {PLDI}. pp. 197--210. {ACM} (2020).
  \doi{10.1145/3385412.3385971}

\bibitem{mazurkiewicz:trace-theory}
Mazurkiewicz, A.W.: Trace theory. In: Advances in Petri Nets. pp. 279--324.
  Lecture Notes in Computer Science, Springer (1986).
  \doi{10.1007/3-540-17906-2\_30}

\bibitem{DBLP:journals/computing/MutluergilT19}
Mutluergil, S.O., Tasiran, S.: A mechanized refinement proof of the {Chase-Lev}
  deque using a proof system. Computing  \textbf{101}(1),  59--74 (2019).
  \doi{10.1007/S00607-018-0635-4}

\end{thebibliography}
\end{document}